\documentclass[a4paper,UKenglish,cleveref,autoref]{lipics-v2021}
\usepackage{booktabs}
\usepackage{xspace}
\usepackage{comment}
\usepackage{thmtools}
\usepackage{booktabs}
\usepackage{todonotes}
\nolinenumbers

\hideLIPIcs  

\usepackage{microtype} 

\usepackage{cite,hyperref}
\hypersetup{
    colorlinks=true,
    linkcolor=blue,
    filecolor=magenta,
    citecolor = red,
    urlcolor=cyan
    }

\graphicspath{ {figures/} }

\newtheorem*{conjecture*}{Conjecture}

\newcommand {\vis} {\text{Vis}}
\newcommand {\calC} {\mathcal{C}}
\newcommand {\calD} {\mathcal{D}}
\newcommand {\calA} {\mathcal{A}}
\newcommand {\calT} {\mathcal{T}}
\newcommand {\calQ} {\mathbb{Q}}

\newcommand {\calI} {\mathcal{I}}
\newcommand {\core} {P_{\text{core}}}

\title{Visibility Queries in Simple Polygons\footnote{A preliminary version of this paper will appear in {\em Proceedings of the 53rd International Colloquium on Automata, Languages, and Programming (ICALP 2026)}.}}

\author{Sujoy Bhore}{Department of Computer Science and Engineering, Indian Institute of Technology Bombay, India}{sujoy@cse.iitb.ac.in}{https://orcid.org/0000-0003-0104-1659}{Work supported in part by ANRF ARG-MATRICS, Grant 002465.}

\author{Chih-Hung Liu}{Department of Electrical Engineering, National Taiwan University, Taiwan}{chliu@ntu.edu.tw}{https://orcid.org/0000-0001-9683-5982}{Supported by Ministry of Education, Taiwan under Yushan Fellow Program with the grant number MOE-111-YSFEE-0003-006-P1.}

\author{Anurag Murty Naredla}{Department of Computer Science, University of Manitoba, Canada}{anurag.naredla@umanitoba.ca}{https://orcid.org/0000-0002-3577-903X}{}

\author{Yakov Nekrich}{Michigan Technological University, USA}{yakov.nekrich@googlemail.com}{https://orcid.org/0000-0003-3771-5088}{Supported by the National Science Foundation under NSF grant 2203278.}

\author{Eunjin Oh}{POSTECH, Korea}{eunjin.oh@postech.ac.kr}{https://orcid.org/0000-0003-0798-2580}{Supported by Institute of Information \& Communications Technology Planning \& Evaluation (IITP) grant funded by the Korea government (MSIT) (No.RS-2024-00440239, Sublinear Scalable Algorithms for Large-Scale Data Analysis) and the National Research Foundation of Korea (NRF) grant funded by the Korea government (MSIT) (No.RS-2024-00358505).}

\author{André van Renssen}{The University of Sydney, Australia}{andre.vanrenssen@sydney.edu.au}{https://orcid.org/0000-0002-9294-9947}{This research was partially funded by the Australian Government through the Australian Research Council (project number DP240101353).}

\author{Frank Staals}{Department of Information and Computing Sciences, Utrecht University, The Netherlands}{f.staals@uu.nl}{https://orcid.org/0009-0004-8522-1351}{}

\author{Haitao Wang}{Kahlert School of Computing, University of Utah, USA}{haitao.wang@utah.edu}{https://orcid.org/0000-0001-8134-7409}{}

\author{Jie Xue}{New York University Shanghai, China}{jiexue@nyu.edu}{0000-0001-7015-1988}{}

\authorrunning{Bhore, Liu, Naredla, Nekrich, Oh, van Renssen, Staals, Wang, and Xue}



\acknowledgements{
This work originated at Shonan Meeting No.~234, “Path Planning and Routing in Geometric Environments.” We thank the participants for their valuable discussions, and we are also grateful to the organizers and the National Institute of Informatics for hosting the event.}

\begin{document}

\maketitle

\begin{abstract}
Given a simple polygon $P$ with $n$ vertices, we consider the problem of constructing a data structure for visibility queries: for any query point $q \in P$, compute the visibility polygon of $q$ in $P$. To obtain $O(\log n + k)$ query time, where $k$ is the size of the visibility polygon of $q$, the previous best result requires $O(n^3)$ space. In this paper, we propose a new data structure that uses $O(n^{2+\epsilon})$ space, for any $\epsilon > 0$, while achieving the same query time.
If only $O(n^2)$ space is available, the best known result provides $O(\log^2 n + k)$ query time. We improve this to $O(\log n \log \log n + k)$ time. 
When restricted to $o(n^2)$ space, the only previously known approach, aside from the $O(n)$-time algorithm that computes the visibility polygon without preprocessing, is an $O(n)$-space data structure that supports $O(k \log n)$-time queries. We construct a data structure using $O(n \log n)$ space that answers visibility queries in $O(n^{1/2+\epsilon} + k)$ time. 
In addition, for the special case in which $q$ lies on the boundary of $P$, we build a data structure of $O(n \log n)$ space supporting $O(\log^2 n + k)$ query time; alternatively, we achieve $O(\log n + k)$ query time using $O(n^{1+\epsilon})$ space. To achieve our results, we propose a new method for decomposing simple polygons, which may be of independent interest.  
\end{abstract}

\section{Introduction}
\label{sec:intro}

Let $P$ be a simple polygon with $n$ vertices. A point $p \in P$ is said to be {\em visible} to a point $q \in P$ if the line segment $\overline{pq}$ lies entirely inside $P$. The {\em visibility polygon} $\vis(q)$ of $q$ is the set of all points in $P$ that are visible to $q$; see Figure~\ref{fig:vispoly}. It is well known that $\vis(q)$ is a star-shaped polygon whose vertices are either:  
(i) vertices of $P$, or (ii) intersection points between edges of $P$ and rays from $q$ passing through reflex vertices of $P$.  

In this paper, we study data structures for {\em visibility queries}: given a query point $q \in P$, compute the visibility polygon $\vis(q)$ efficiently.

\subparagraph{Previous work.}
$\vis(q)$ can be computed in $O(n)$ time without any preprocessing~\cite{ref:ChazelleVi89,ref:GuibasLi87,ref:ElGindyA81}. If $O(\log n)$-time ray-shooting queries are available, then $\vis(q)$ can be computed in $O(k \log n)$ time~\cite{ref:AronovVi02}, where $k$ is the size of $\vis(q)$. Such ray-shooting data structures can be constructed in $O(n)$ time~\cite{ref:ChazelleVi89,ref:ChazelleRa94,ref:HershbergerA95}.
Ideally, one would like a query algorithm whose running time is linear in the output size, i.e., $O(k)$. Data structures achieving this are known, but at the cost of larger space. Bose, Lubiw, and Munro~\cite{ref:BoseEf02}, and independently Guibas, Motwani, and Raghavan~\cite{ref:GuibasTh92}, obtained $O(\log n + k)$ query time using $O(n^3)$ space and $O(n^3 \log n)$ preprocessing time. Aronov, Guibas, Teichmann, and Zhang~\cite{ref:AronovVi02} constructed a data structure of $O(n^2)$ space in $O(n^2 \log n)$ time that supports $O(\log^2 n + k)$-time visibility queries.
Note that $O(\log n + k)$ is the optimal query time~\cite{ref:BoseEf02,ref:GuibasTh92}.

\begin{figure}[t]
\centering
\includegraphics[height=1.2in]{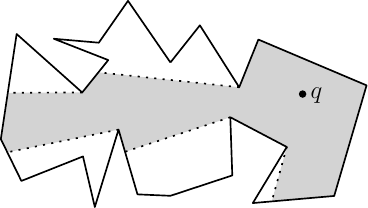}
\caption{The grey region is $\vis(q)$. Each edge of it is either a portion of an edge of $P$ or a dotted segment (which is called a {\em window} in the literature~\cite{ref:BoseEf02}).}
\label{fig:vispoly}
\end{figure}

\subparagraph{Our results.}
There has been no progress on this problem for almost three decades since the conference version of \cite{ref:AronovVi02} was published in 1998. In this paper, we present the following new results. 
\begin{itemize}
\item 
First, we give the first-known subquadratic-space structure that achieves a query time linear in the output size: a data structure of $O(n \log n)$ space, $O(n \log^2 n)$ preprocessing time, and $O(n^{1/2+\epsilon} + k)$ query time. Throughout the paper, $\epsilon$ represents an arbitrarily small positive constant. Furthermore, we obtain the following trade-offs between the query time and preprocessing: for a parameter $r$ satisfying $1 \le r \le n^{1-\epsilon}$, we can build a data structure of space $O(nr + n\log (n/r))$, $O((nr + n\log (n/r))\log n)$ preprocessing time, and $O((n/r)^{1/2+\epsilon}+k)$ query time. Note that this result implies that for a large $k$ (e.g., $k=\Omega(n^{\epsilon})$), we can achieve $O(k)$ query time with sub-quadratic preprocessing. 

\item
Second, to achieve the optimal $O(\log n + k)$ query time, we construct a data structure of $O(n^{2+\epsilon})$ space and preprocessing time, for any $\epsilon > 0$. This improves the space bound of the previous $O(n^3)$-space structures~\cite{ref:BoseEf02,ref:GuibasTh92} by nearly a linear factor. 

\item 
Third, under an $O(n^2)$ space budget, we improve the $O(\log^2 n + k)$ query time in \cite{ref:AronovVi02} to $O(\log n \log\log n + k)$, with a preprocessing time of $O(n^2 \log n)$. 

We remark that even a near-logarithmic improvement in query time is often considered significant. Indeed, for the closely related ray-shooting problem, reducing the query time from $O(\log^2 n)$ to $O(\log n)$ required substantial technical effort~\cite{ref:ChazelleVi89,ref:ChazelleRa94,ref:HershbergerA95}.
\end{itemize}

Table~\ref{tab:results} summarizes these bounds alongside prior work.

\begin{table}[h]
    \caption{Summary of the results for visibility queries ($k$ is the output size).}
    \label{tab:results}
    \centering
    \tabcolsep1ex%
    \vspace{0.15in}
    {%
        \begin{tabular}{l|lll}
         \toprule
            Results & Query time & Space & Preprocessing time  \\
            \hline
            \cite{ref:ChazelleRa94,ref:HershbergerA95} & $O(k\log n)$ &  $O(n)$ &  $O(n)$\\            
            \cite{ref:BoseEf02,ref:GuibasTh92} & $O(\log n + k)$ &  $O(n^3)$ &  $O(n^3\log n)$\\            
            \cite{ref:AronovVi02} & $O(\log^2 n + k)$ &  $O(n^2)$ &  $O(n^2\log n)$\\            
            Ours (near-linear space) & $O(n^{1/2+\epsilon} + k)$ &  $O(n\log n)$ &  $O(n\log^2 n)$\\         
            Ours (optimal query time) & $O(\log n + k)$ &  $O(n^{2+\epsilon})$ &  $O(n^{2+\epsilon})$\\            
            Ours (quadratic space) & $O(\log n\log\log n + k)$ &  $O(n^2)$ &  $O(n^2\log n)$\\                                   
            \bottomrule
        \end{tabular}}
\end{table}

In addition, we study a natural special case in which the query point $q$ lies on the boundary of $P$. To the best of our knowledge, this {\em boundary case} has not been treated separately in the literature. We obtain a data structure of $O(n \log n)$ space and $O(n \log^2 n)$ preprocessing time that answers queries in $O(\log^2 n + k)$ time. Alternatively, we can achieve $O(\log n + k)$ query time using $O(n^{1+\epsilon})$ space and preprocessing time.

\subparagraph{Other related work.}
Several related visibility query problems have also been studied. Chen and Daescu~\cite{ref:ChenMa98} considered a {\em segment-restricted} version, where the query point $q$ lies on a given segment in $P$. Their data structure uses $O(n)$ space and supports queries in $O(\log n + k)$ time. However, their approach does not appear to extend to arbitrary query points.

Given a segment $s \subseteq P$, a point is {\em weakly visible} to $s$ if it is visible to at least one point of $s$. The {\em weak visibility polygon} of $s$ consists of all points of $P$ visible to $s$, and can be computed in $O(n)$ time~\cite{ref:GuibasLi87}. The corresponding weak visibility query problem, i.e., computing this polygon for an arbitrary query segment, has been investigated in several works~\cite{ref:ChenWe15,ref:AronovVi02,ref:BygiWe11,ref:BoseEf02}. Using $O(n)$ space and preprocessing time, one can support queries in $O(k \log n)$ time; alternatively, using $O(n^3)$ space and preprocessing time, queries can be answered in $O(\log n + k)$ time~\cite{ref:ChenWe15}.

As a special visibility problem, the {\em ray-shooting} problem asks for the first point on $\partial P$ hit by a ray originating inside $P$. As mentioned earlier, ray-shooting queries can be answered in $O(\log n)$ time after $O(n)$-time preprocessing~\cite{ref:ChazelleRa94,ref:HershbergerA95}.

Visibility and ray-shooting in polygons with holes have also been studied~\cite{ref:AgarwalRa96,ref:ChenVi15,ref:InkuluVi09,ref:LuPo11,ref:PocchiolaGR90,ref:ZareiQu08}, but the presence of holes makes the problems substantially more difficult. We refer the reader to \cite{ref:ChenVi15} for an overview of results in that setting.

\subsection{An overview of our approach}
\label{sec:overview}

The methods of \cite{ref:BoseEf02,ref:GuibasTh92} decompose $P$ into $O(n^3)$ cells such that the combinatorial structure of $\vis(q)$ is the same for every point $q$ in the same cell. 
Since the decomposition itself has cubic size, these approaches cannot be directly adapted when one wishes to use subcubic space.

The approach \cite{ref:AronovVi02} instead builds on the {\em balanced decomposition} of $P$~\cite{ref:GuibasLi87,ref:GuibasOp89,ref:ChazelleVi89} (see Section~\ref{sec:pre} for more details). 
Their method achieves $O(\log^2 n + k)$ query time using $O(n^2)$ space. The $O(\log^2 n)$ factor in query time appears inherent to the balanced decomposition: each subpolygon may be connected to the ``outside world'', i.e., the rest of $P$, by $O(\log n)$ diagonals and the decomposition has $O(\log n)$ recursive levels.

To obtain $O(\log n + k)$ query time with subcubic space, we therefore require a new strategy. Our main idea is to introduce a {\em new polygon decomposition} with a crucial structural feature (which we call {\em the gate property}): every subpolygon connects to the outside world through a {\em single} diagonal, which we call its {\em gate}. As in the balanced decomposition, we recursively divide $P$ into subpolygons whose sizes shrink geometrically; however, the gate property fundamentally changes how visibility information propagates across subregions. Note that such a gate property alone might be achieved through other decompositions. However, a key advantage of our decomposition is that it enables efficient preprocessing, allowing certain types of visibility queries to be answered efficiently (see Theorem~\ref{theo:treequeries}).

Intuitively, suppose a query point $q$ lies in a subpolygon $P_d$ with gate $d$. If we have already computed the portion of $P_d$ visible to $q$, denoted by $\vis(P_d,q)$, then the remaining part of the visibility region, $\vis(P \setminus P_d,q)$, is exactly the portion of $P \setminus P_d$ visible through the diagonal $d$. We refer to computing this region as the {\em diagonal-separated subproblem}. The entire query algorithm reduces to solving this subproblem and recursing inside $P_d$ (i.e., to compute $\vis(P_d,q)$, we can apply the algorithm recursively).

Our optimal-query-time data structure is built on this decomposition and uses $O(1)$ recursive steps by choosing a parameter $r = n^{\epsilon}$. This yields $O(n^{2+\epsilon})$ space and $O(\log n + k)$ query time. Our quadratic-space data structure also relies on our new decomposition but groups its levels into $O(\log\log n)$ {\em super-levels}. This reduces the space to $O(n^2)$ at the cost of increasing the query time to $O(\log n \log\log n + k)$.
In contrast, the balanced decomposition of \cite{ref:AronovVi02} has $\Theta(\log n)$ levels, resulting in an unavoidable $O(\log^2 n)$ factor.

Our near-linear space data structure follows the high-level strategy of \cite{ref:AronovVi02} but solves the diagonal-separated subproblem more efficiently. While \cite{ref:AronovVi02} used an $O(n^2)$-space data structure supporting $O(\log n + k)$ queries for this subproblem, we instead reduce the problem to simplex stabbing queries and adapt the recent structure of Chan and Zhang~\cite{ref:ChanSi23}. This gives an $O(n)$-space structure supporting queries in $O(n^{1/2+\epsilon} + k)$ time, yielding our result.

For the boundary case, the diagonal-separated subproblem simplifies substantially: with $q$ constrained to lie on the boundary of $P$, the remaining computation becomes essentially one-dimensional. We obtain a solution of $O(n)$ space and $O(\log n + k)$ query time. Plugging this into the balanced decomposition method of \cite{ref:AronovVi02}, we obtain our first data structure for the boundary case. Plugging this into our new decomposition yields our second data structure. 

It is worth noting that our new decomposition may be interesting in its own right, and we believe it will find additional applications in the future.

\subparagraph{Outline.}
In the following, after introducing notation in Section~\ref{sec:pre}, we present our near-linear space data structure in Section~\ref{sec:smallspace}. In Section~\ref{sec:decomp}, we describe our new polygon decomposition, which serves as a key component for the optimal-query-time data structure in Section~\ref{sec:Optimal_query_time} and the quadratic-space solution in Section~\ref{sec:A_Quadratic_Space_Structure}. The boundary case is discussed in Section~\ref{sec:Boundary_queries_only}.


\section{Preliminaries}
\label{sec:pre}
We follow the notation introduced in Section~\ref{sec:intro}, e.g., $P$, $n$, and $\vis(q)$. For any subpolygon $P' \subseteq P$, we define $\vis(P',q) = \vis(q) \cap P'$, i.e., the portion of $\vis(q)$ contained in $P'$. 

For convenience, we assume that each edge of $P$ is an open segment that does not include its endpoints; thus, the vertices of $P$ are disjoint from its edges. We refer to either a vertex or an edge of $P$ as an {\em element} of $P$. For any geometric object $R$, let $\partial R$ denote its boundary. 
For two points $p$ and $q$, $\overline{pq}$ denotes the line segment connecting them. 
 
It is well known that $\vis(q)$ is star-shaped and can be represented by a cyclic (say, clockwise) list of its vertices. Alternatively, given the cyclic list of vertices and edges of $P$ that appear on the boundary $\partial \vis(q)$, one can reconstruct $\vis(q)$ in $O(|\vis(q)|)$ time~\cite{ref:AronovVi02,ref:BoseEf02}. Hence, to compute $\vis(q)$, it suffices to determine this cyclic list, called the {\em combinatorial representation} of $\vis(q)$~\cite{ref:AronovVi02}. It is also known that this cyclic order of vertices and edges of $P$ on $\partial \vis(q)$ is consistent with their order along $\partial P$. In the following, depending on the context, we may use $\vis(q)$ to refer to its combinatorial representation.

\subparagraph{Balanced decomposition.}
We will use the {\em balanced decomposition} of $P$, which has been applied to various problems in simple polygons~\cite{ref:AronovVi02,ref:GuibasLi87,ref:GuibasOp89,ref:ChazelleVi89}. A {\em diagonal} is a line segment connecting two vertices of $P$ that is fully contained in $P$ but not an edge of $P$. 
Chazelle~\cite{ref:ChazelleA82} proved that $P$ always admits a diagonal that splits it into two subpolygons, each having between $n/3$ and $2n/3$ vertices. The balanced decomposition of $P$ is obtained by recursively partitioning each subpolygon using such diagonals until all subpolygons are triangles. This decomposition can be represented by a tree $\calT(P)$, called the {\em BD-tree}. Each node $v$ of $\calT(P)$ corresponds to a subpolygon $P_v$ of $P$: the root corresponds to $P$ itself, and each leaf corresponds to a triangle. The diagonal used to partition $P_v$ into its two child subpolygons is stored at $v$ and denoted by $d_v$. The triangles at all leaves of $\calT(P)$ form a triangulation of $P$. The entire decomposition and BD-tree $\calT(P)$ can be computed in $O(n)$ time~\cite{ref:GuibasOp89,chazelle91triangulation}. 

\subparagraph{Convention on stating query time.}
For simplicity, when we state that a data structure supports queries in $O(Q(n))$ time for computing $\vis(q)$ (or a portion thereof, e.g., $\vis(P',q)$ for a subpolygon $P' \subseteq P$), we mean that a (balanced) binary search tree storing the combinatorial representation of $\vis(q)$ can be computed in $O(Q(n))$ time, and the full visibility polygon $\vis(q)$ can then be output in an additional $O(|\vis(q)|)$ time. 

For conciseness, we say that a data structure has complexity $O(T(n), S(n), Q(n))$ if its preprocessing time, space, and query time are $O(T(n))$, $O(S(n))$, and $O(Q(n))$, respectively.


\section{A near-linear space data structure}
\label{sec:smallspace}

In this section, we present our near-linear space data structure. Our algorithm follows the general framework of \cite{ref:AronovVi02}, but we address a key component, the diagonal-separated subproblem, in a different way, as discussed in Section~\ref{sec:overview}. This subproblem arises in essentially all of our subsequent data structures. The same framework will also be partially used in our quadratic-space data structure in Section~\ref{sec:A_Quadratic_Space_Structure} and in the boundary case solution in Section~\ref{sec:Boundary_queries_only}. Therefore, the discussions in this section are also instrumental to the later parts of the paper.

We first discuss the diagonal-separated subproblem in Section~\ref{sec:diasep}, present its linear-space solution in Section~\ref{sec:coneproblem}, and finally describe the general algorithmic framework in Section~\ref{sec:bdscheme}.

\subsection{A diagonal-separated subproblem}  
\label{sec:diasep}
Let $d$ be a diagonal that divides $P$ into two subpolygons $P_\ell$ and $P_r$. Without loss of generality, assume that $d$ is vertical and $P_\ell$ lies locally to the left of $d$. The {\em diagonal-separated subproblem} is to compute $\vis(P_\ell, q)$ for any query point $q \in P_r$. We describe our solution below. 

Since $q \in P_r$, for any point $p\in \vis(P_\ell,q)$, $\overline{pq}$ must cross $d$. In other words, the visibility of $q$ within $P_\ell$ occurs ``through'' $d$. Because $P$ is a simple polygon, $\vis(d, q)$ is a contiguous segment of $d$~\cite{ref:AronovVi02}. Let $s = \vis(d, q)$. If $s = \emptyset$, then $\vis(P_\ell, q) = \emptyset$. Otherwise, let $C_q(s)$ denote the {\em cone} formed by the rays from $q$ through all points $p \in s$; we say that $C_q(s)$ is {\em defined} by $q$ and $s$, and refer to $q$ as the {\em apex} of the cone (see Figure~\ref{fig:viscone00}). 

Using the two-point shortest path data structure of Guibas and Hershberger~\cite{ref:GuibasOp89} (henceforth referred to as the {\em GH data structure}), the two bounding rays of $C_q(s)$ can be computed in $O(\log n)$ time via shortest-path queries. To compute $\vis(P_\ell, q)$, it thus suffices to determine the portion of $P_\ell$ visible to $q$ through the cone $C_q(s)$.


\begin{figure}[t]
\centering
\includegraphics[height=1.2in]{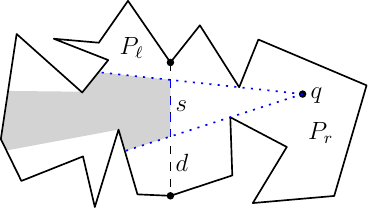}
\caption{The gray region is $\vis(P_\ell,q)$. $C_s(q)$ is bounded by the two blue dotted segments.}
\label{fig:viscone00}
\end{figure}

We use $\vis_s(q)$ to denote $\vis(P_\ell, q)$. We also define an analogous notion $\vis_d(q)$, which consists of all points $p \in P_\ell$ such that $\overline{pq}$ crosses $d$ and $\overline{pp'}\subseteq P_\ell$, where $p' = d \cap \overline{pq}$. Note that whether the segment $\overline{qp'}$ lies inside $P$ is irrelevant. See Figure~\ref{fig:viscone10}.

\begin{figure}[t]
\centering
\includegraphics[height=1.2in]{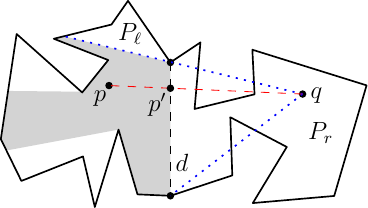}
\caption{The gray region is $\vis_d(q)$.}
\label{fig:viscone10}
\end{figure}

Following \cite{ref:AronovVi02}, we first compute a binary search tree that stores the combinatorial representation of $\vis_d(q)$, i.e., the cyclic order of the elements of $P_\ell$ (recall that an element refers to either a vertex or an edge of $P$) appearing on $\partial \vis_d(q)$. We then perform a {\em clipping} operation on this tree using $C_q(s)$; the resulting tree represents $\vis_s(q)$. 


\begin{observation}\label{obs:boundaryorder}{\em \cite{ref:AronovVi02,ref:BoseEf02}}
The order of the vertices and edges of $P_\ell$ along $\partial \vis_s(q)$ (resp., $\partial \vis_d(q)$) is consistent with their order along $\partial P_\ell$.
\end{observation}

We sort all elements of $P_\ell$ along its boundary and assign each of them an {\em index} in this order. 
Let $A_d(q)$ denote the set of elements of $P$ that appear on $\partial \vis_d(q)$. 
By Observation~\ref{obs:boundaryorder}, the elements of $A_d(q)$ are sorted in index order in the combinatorial representation of $\vis_d(q)$. 

\begin{figure}[t]
\centering
\includegraphics[height=1.5in]{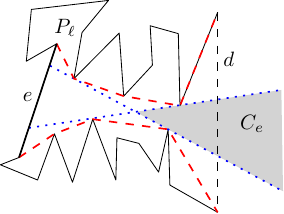}
\caption{The gray region is $C_e$. The two dashed red paths are shortest paths from the two endpoints of $e$ to the two vertices of $d$. The two bounding lines of $C_e$ are tangent to these two paths.}
\label{fig:edgecone}
\end{figure}


For each vertex $v \in P_\ell$, as discussed above, all points of $d$ visible to $v$ within $P_\ell$ form a segment $s_v$ of $d$. Let $C_v$ denote the cone defined by $s_v$ and $v$. Observe that $v$ belongs to $A_d(q)$ if and only if $q \in C_v$. 
For each edge $e$ of $P_\ell$, as discussed in \cite{ref:GuibasLi87}, the points of $d$ weakly visible to $e$ also form a segment $s_e$ of $d$, such that $e$ belongs to $A_d(q)$ if and only if $q$ lies in the cone $C_e$ defined by $s_e$ and a point to the left of $d$ (see Figure~\ref{fig:edgecone}).

For reference, we summarize the above discussion in the following observation.

\begin{observation}\label{obs:cones}
For each vertex $v$ (resp., edge $e$) of $P_\ell$, $v$ (resp., $e$) appears in the combinatorial representation of $\vis_d(q)$ if and only if $q$ lies in the cone $C_v$ (resp., $C_e$). 
\end{observation}

All cones $C_v$ and $C_e$ for the vertices and edges of $P_\ell$ can be computed in $O(n)$ time using the shortest path trees of the two endpoints of $d$~\cite{ref:GuibasLi87}. 
Let $\mathcal{C}$ denote the set of all these cones. We refer to $v$ (resp., $e$) as the {\em generator} of $C_v$ (resp., $C_e$). 
To compute $A_d(q)$, by Observation~\ref{obs:cones}, the problem reduces to the following {\em cone stabbing query}: given $q$, compute all cones of $\mathcal{C}$ that are stabbed by $q$ (we say that a cone is {\em stabbed} by $q$ if the cone contains $q$). To accommodate the clipping operation with $C_q(s)$, we require the data structure to return a binary search tree storing the generators of the stabbed cones in their index order.

\subparagraph{A quadratic-space solution.}
A quadratic-space solution is given in \cite{ref:AronovVi02}. We briefly review it below as it also serves as a subroutine in some of our later algorithms.

Let $R_d$ denote the halfplane to the right of the supporting line of $d$. 
If $q \notin R_d$, then $\vis_d(q) = \emptyset$, so we assume $q \in R_d$. 
Let $\mathcal{A}_d$ be the arrangement in $R_d$ formed by the supporting lines of the bounding rays of all cones in $\mathcal{C}$.  
By Observation~\ref{obs:cones}, points $q$ in the same cell of $\mathcal{A}_d$ share the same (combinatorial representation of) $\vis_d(q)$, and $\vis_d(q)$ of $q$ in adjacent cells differ by at most one element of $P_\ell$.  
Since there are $O(n)$ cones, $\mathcal{A}_d$ has $O(n^2)$ cells.  

We use persistent binary search trees~\cite{ref:SnarnakPl86,ref:DriscollMa89} 
to store the combinatorial representations of $\vis_d(q)$ for all cells of $\mathcal{A}_d$ 
in a total of $O(n^2)$ space and $O(n^2 \log n)$ preprocessing time.  
Additionally, we construct a point location data structure on $\mathcal{A}_d$ in 
$O(n^2)$ time and space~\cite{ref:KirkpatrickOp83,ref:EdelsbrunnerOp86}.  
Given a query point $q$, we can locate the cell of $\mathcal{A}_d$ containing $q$ in $O(\log n)$ time and thereby obtain access to the corresponding binary search tree that stores $\vis_d(q)$.  
This yields a data structure of complexity  $O(n^2 \log n, n^2, \log n)$, i.e., $O(n^2 \log n)$ preprocessing time, 
$O(n^2)$ space, and  $O(\log n)$ query time for computing $\vis(P_\ell, q)$ for any $q \in P_r$.



\subparagraph{The boundary case.} In the boundary case, each query point $q$ is on $\partial P\cap P_r$. In this case the cone stabbing queries become 1D {\em interval stabbing} queries on $\partial P\cap P_r$. Using persistent trees~\cite{ref:SnarnakPl86,ref:DriscollMa89}, 
we obtain a data structure of complexity $O(n\log n, n, \log n)$. This result combining with the algorithmic framework given in Section~\ref{sec:bdscheme} yields a data structure of complexity $O(n\log^2 n, n\log n, \log^2 n)$ for computing $\vis(q)$. 
See Section~\ref{sec:Boundary_queries_only} for details. 

\subsection{A linear-space solution to the diagonal-separated subproblem}
\label{sec:coneproblem}


We now present our solution to the cone stabbing problem and thus solve the diagonal-separated subproblem. 
\subparagraph{Overview.}
Our algorithm computes a collection of ``canonical subsets'' of $\calC$, where each subset is represented by a binary search tree storing the cone generators in their index order. We perform a clipping operation on each tree using the cone $C_s(q)$. The elements remaining in the clipped trees of all canonical subsets collectively form the set $A_s(q)$, which consists of the elements of $P$ lying on $\partial \vis_s(q)$. We then sort the elements of $A_s(q)$ by their index order to obtain $\vis_s(q)$. Finally we build a binary search tree to store them and return this tree as the answer to the query. The details are given below. 
\bigskip



We adapt Chan and Zhang's data structure~\cite[Section 3.1]{ref:ChanSi23} for simplex stabbing queries (given a query point $q$, compute all simplices stabbed by $q$ among a set of given simplices; each of our cones is a special simplex in 2D). 
The data structure has multiple levels and it is constructed using Matou\v{s}ek's simplicial partition~\cite{ref:MatousekEf92}. To solve our problem, we can build the data structure on the cones of $\calC$ with $j=2$ levels. We briefly discuss this and the reader is referred to \cite[Section 3.1]{ref:ChanSi23} for details. 

Determining whether $q$ is in a cone is equivalent to testing whether $q$ is in the intersection of two halfplanes bounded by the cone's bounding lines. In the dual plane, this is to check whether the dual line $q^*$ of $q$ lies in the ``correct'' side of the two dual points of the bounding lines of the cone. We apply Matou\v{s}ek's simplicial partition~\cite{ref:MatousekEf92} on the dual points of the $j$-th halfplanes of the cones of $\calC$, with $j=2$. For each cell $\Delta$ of the partition that intersects $q^*$, we recurse on the canonical subset of the dual points of $\Delta$. For each cell $\Delta$ lying entirely below $q^*$ (or above $q^*$, depending on which is the correct side), we recurse on the canonical subset of $\Delta$ but as a level-($j-1$) problem. For each canonical subset in a level-0 problem, we build a binary search tree to store all generators of cones in the subset. Following the analysis in \cite[Section~3.1]{ref:ChanSi23}, the resulting data structure is of $O(n)$ space and can be constructed in $O(n\log n)$ time. 

Given a query point $q$, the query algorithm returns $O(n^{1/2+\epsilon})$ canonical subsets whose union corresponds to the set of cones of $\calC$ stabbed by $q$. For each canonical subset, we perform a clipping operation using the cone $C_q(s)$, taking $O(\log n)$ additional time. Hence, we can compute $A_s(q)$ in $O(n^{1/2+\epsilon}\log n+|A_s(q)|)$ time, which simplies to  $O(n^{1/2+\epsilon}+|A_s(q)|)$ since the $\log n$ factor is absorbed by $n^{\epsilon}$. 

\subparagraph{Remark.}
Chan and Zhang also proposed a more efficient data structure that, in the two-dimensional case, answers simplex range stabbing \emph{counting} queries in $O(\sqrt{n})$ time~\cite[Section~3]{ref:ChanSi23}. However, that structure is designed for the group setting and relies on subtraction operations to compute the number of simplices stabbed by a query point. In our setting, where clipping operations are required, this subtraction-based approach is no longer applicable.  
Furthermore, Chan and Zhang~\cite[Section~3.5]{ref:ChanSi23} also considered simplex range stabbing \emph{reporting} queries and developed a data structure that, in the two-dimensional case, answers each reporting query in $O(\sqrt{n} + k)$ time, where $k$ denotes the output size. One might hope to adapt this structure to improve our cone stabbing queries, and specifically, to reduce the $O(n^{1/2+\epsilon})$ factor to $O(\sqrt{n})$. However, since their reporting query bound also depends on the output size, the $O(\sqrt{n})$ term does not necessarily bound the number of canonical subsets. In our problem, where clipping operations are essential, we require a genuine upper bound on the number of canonical subsets. Hence, adapting their reporting data structure to our setting appears challenging.
\medskip
\medskip

Once $A_s(q)$ is obtained, we must sort its elements by their index order. A straightforward comparison sort would require $O(|A_s(q)|\log |A_s(q)|)$ time, introducing an undesired logarithmic factor in the overall query time. However, since the index of each element of $A_s(q)$ is an integer in $[1,2n]$, we can perform the sorting in $O(n^{\epsilon} + |A_s(q)|)$ time using radix sort in the following lemma. 

\begin{lemma}\label{lem:sort}
For any $\epsilon>0$, the elements of $A_s(q)$ can be sorted by their indices in $O(n^{\epsilon}+|A_s(q)|)$ time.
\end{lemma}
\begin{proof}
Let $I$ denote the set of indices of the elements in $A_s(q)$, and our goal is to sort the integers in $I$.  
Since each index in $I$ is an integer in $[1,2n]$, it can be represented using at most $\log n + 1$ bits. We partition these $\log n + 1$ bits into $O(1)$ groups, each consisting of $\epsilon \log n$ bits. We refer to each such group as a {\em digit}. Hence, each integer in $I$ has $O(1)$ digits.  
We apply radix sort on the integers of $I$ using these digits. Sorting $I$ by a single digit requires $O(2^{\epsilon \log n} + |I|) = O(n^{\epsilon} + |I|)$ time. Since there are only $O(1)$ digits, the entire sorting process takes $O(1)$ passes, resulting in a total running time of $O(n^{\epsilon} + |I|)$.
\end{proof}

After $A_s(q)$ is sorted, we obtain the combinatorial representation of $\vis_s(q)$. We then construct a binary search tree to store this representation in an additional $O(|\vis_s(q)|)$ time. Combining all the above results yields the following lemma.  

\begin{lemma}\label{lem:linearspace}
Given a simple polygon $P$ of $n$ vertices and a diagonal $d$ dividing $P$ into two subpolygons $P_\ell$ and $P_r$, 
a data structure of $O(n)$ space can be constructed in $O(n\log n)$ time such that $\vis(P_\ell,q)$ can be computed in $O(n^{1/2+\epsilon}+|\vis(P_\ell,q)|)$ time for any point $q\in P_r$. 
\end{lemma}

\subsection{A general algorithmic framework}
\label{sec:bdscheme}


We now describe the algorithmic framework. In the preprocessing, we compute a balanced decomposition of $P$ together with the BD-tree $\calT(P)$ as defined in Section~\ref{sec:pre}.

Given a query point $q \in P$, let $v_q$ be the leaf of $\calT(P)$ whose triangle $P_{v_q}$ contains $q$. The query algorithm proceeds in a bottom-up manner along the path from $v_q$ to the root of $\calT(P)$. For each node $v$ on this path, let $u$ denote the parent of $v$ and $w$ the sibling of $v$ (i.e., the other child of $u$). Suppose that $\vis(P_v,q)$ has already been computed (this is true initially when $v = v_q$ since $\vis(P_v,q) = P_{v_q}$). Our goal at node $u$ is to compute $\vis(P_u,q)$, which equals the union of $\vis(P_v,q)$ and $\vis(P_w,q)$.  

By definition, $P_u$ is partitioned into $P_v$ and $P_w$ by the diagonal $d_u$. Since $\vis(P_v,q)$ is known, we next compute $\vis(P_w,q)$ and then merge $\vis(P_v,q)$ and $\vis(P_w,q)$ along $d_u$ to obtain $\vis(P_u,q)$. Computing $\vis(P_w,q)$ is precisely an instance of the diagonal-separated subproblem. Therefore, in the preprocessing, for each node $u \in \calT(P)$, we construct a diagonal-separated data structure for both $P_v$ and $P_w$ (with respect to $d_u$). During a query, these data structures allow us to compute $\vis(P_w,q)$ efficiently.



Once $\vis(P_w,q)$ is computed, we obtain $\vis(P_u,q)$ by merging $\vis(P_w,q)$ and $\vis(P_v,q)$ along $s = \vis(d_u,q)$. Note that if $s = \emptyset$, then $\vis(P_w,q) = \emptyset$; otherwise, $s$ must be an edge of $\vis(P_v,q)$. Merging $\vis(P_v,q)$ and $\vis(P_w,q)$ can be done in $O(\log n)$ time. This merging step was briefly mentioned in~\cite{ref:AronovVi02}; for completeness and reference purpose, we provide some details in Lemma~\ref{lem:merge}. 
After computing $\vis(P_u,q)$, the algorithm continues upward until reaching the root of $\calT(P)$, at which point $\vis(q)$ is obtained.

\begin{lemma}\label{lem:merge}
Suppose a diagonal $d$ divides $P$ into two subpolygons $P_\ell$ and $P_r$. Let $q$ be a point in $P_r$. 
Given a binary search tree storing $\vis(P_r,q)$ and a binary search tree storing $\vis(P_\ell,q)$, we can obtain a binary search tree storing $\vis(q)$ in $O(\log n)$ time. 
\end{lemma}
\begin{proof}
Let $T_r$ be the tree storing $\vis(P_r,q)$ and $T_\ell$ the tree storing $\vis(P_\ell,q)$. 
Let $s=\vis(d,q)$.
If $s=\emptyset$, then $\vis(P_\ell, q)=\emptyset$ and thus $\vis(q)=\vis(P_r,q)$. In this case, we can simply return $T_r$. If $s\neq \emptyset$, then $d$ must appear as an element in $\vis(P_r,q)$. We assume that $T_r$ stores $\vis(P_r,q)$ in such a way that each element of $\vis(P_r,q)$ is stored in a leaf of $T_r$. Let $v_{d}$ be the leaf of $T_r$ storing $d$. 

The visibility of $q$ in $P_\ell$ is constrained by the cone $C_q(s)$ defined by $q$ and $s$, i.e., $\vis(P_\ell,q)\subseteq C_q(s)$. Therefore, if we replace the leaf $v_{d}$ in $T_r$ by the entire tree $T_\ell$, then we could obtain a new tree that exactly stores the combinatorial representation of $\vis(q)$. However, the new tree thus obtained might not be balanced. To achieve a balanced tree, we instead do the following. We first perform a split operation to split $T_r$ into two trees along the path from the root to $v_{d}$ (and also remove $v_{d}$). After that, we perform a splice operation to splice the above two trees with $T_\ell$ by having $T_\ell$ in the middle. These split and splice operations can be done in $O(\log n)$ time. The new tree thus obtained is balanced. We finally return the root of the new tree. The total time is $O(\log n)$.
\end{proof}

If we plug the quadratic-space diagonal-separated data structure of~\cite{ref:AronovVi02} (also reviewed in Section~\ref{sec:diasep}) into the framework, then, since the sizes of the subpolygons decrease geometrically along any root-to-leaf path of $\calT(P)$, we obtain a data structure of complexity $O(n^2\log n, n^2, \log^2 n)$ for computing $\vis(q)$ for any $q \in P$. This is the result of~\cite{ref:AronovVi02}.  

If we instead apply our linear-space result from Lemma~\ref{lem:linearspace}, then, because the total size of all subpolygons at each level of $\calT(P)$ is $O(n)$ (as they form a partition of $P$), we obtain a data structure of complexity $O(n\log^2 n, n\log n, n^{1/2+\epsilon})$. 

\subparagraph{Trade-off.}
We can obtain a trade-off by applying a somewhat standard trade-off technique based on hierarchical cuttings~\cite{ref:ChazelleCu93}. Specifically, we first have the following trade-off of Lemma~\ref{lem:linearspace} for the  diagonal-separated problem.

\begin{lemma}\label{lem:tradeoff}
Let $P$ be a simple polygon with $n$ vertices, and let $d$ be a diagonal that divides $P$ into two subpolygons $P_\ell$ and $P_r$.  
Suppose $r$ is a parameter satisfying $1 \le r \le n^{1-\epsilon}$ for any constant $\epsilon > 0$.  
Then, a data structure of size $O(nr)$ can be constructed in $O(nr \log (n/r))$ time such that $\vis(P_\ell, q)$ can be computed in $O((n/r)^{1/2+\epsilon} + |\vis(P_\ell, q)|)$ time for any query point $q \in P_r$.
\end{lemma}
\begin{proof}
We first introduce the hierarchical cutting~\cite{ref:ChazelleCu93} and then explain our approach. 

Let $H$ be a set of $n$ lines in the plane. For a parameter $r$ with $1 \leq r \leq n$, a {\em $(1/r)$-cutting} $\Xi$ for $H$ is a collection of interior-disjoint triangles (called {\em cells}) whose union covers the entire plane such that $|H_{\sigma}| \leq n/r$, where $H_{\sigma}$ is the subset of lines of $H$ crossing the interior of $\sigma$ ($H_{\sigma}$ is called the {\em conflict list} of $\sigma$). The {\em size} of $\Xi$ is the number of cells of $\Xi$. 

A cutting $\Xi'$ \emph{$c$-refines} another cutting $\Xi$ if each cell of $\Xi'$ is contained in a cell of $\Xi$ and each cell of $\Xi$ contains at most $c$ cells of $\Xi'$. 
A {\em hierarchical $(1/r)$-cutting} for $H$ consists of a sequence of cuttings $\Xi_0, \Xi_1, ..., \Xi_t$, in which $\Xi_0$ has only one cell that is the entire plane, and each $\Xi_i$, $1 \leq i \leq t$, is a $(1/\rho^i)$-cutting of size $O(\rho^{2i})$ that $c$-refines $\Xi_{i - 1}$, for two constants $\rho$ and $c$. By setting $t = \lceil \log_\rho r \rceil$, the last cutting $\Xi_t$ is a $(1/r)$-cutting. If a cell $\sigma'$ of $\Xi_{i - 1}$  contains a cell $\sigma$ of $\Xi_i$, we say that $\sigma'$ is the \emph{parent} of $\sigma$ and $\sigma$ is a \emph{child} of $\sigma'$.  
As such, the hierarchical $(1/r)$-cutting forms a tree structure with the single cell of $\Xi_0$ as the root. Let $\Xi$ denote the set of all cells in all cuttings $\Xi_i$, $0\leq i\leq t$.
The total number of cells of $\Xi$ is $O(r^2)$ and the sum of the conflict sizes of all cells of $\Xi$ is $O(nr)$~\cite{ref:ChazelleCu93}.
A hierarchical $(1/r)$-cutting for $H$ can be computed in $O(nr)$ time, along with the conflict lists $H_{\sigma}$ for all cells $\sigma\in \Xi$~\cite{ref:ChazelleCu93}. 

We next prove the lemma. Let $H$ denote the set of the supporting lines of bounding rays of the cones of $\calC$. 

\subparagraph{Preprocessing.}
In the preprocessing, we construct a hierarchical $(1/r)$-cutting for $H$ as above, for some parameter $1\leq r\leq n^{1-\epsilon}$. For each $1\leq i\leq t$, for each cell $\sigma\in \Xi_i$, 
we compute $\calC(\sigma)$, the set of cones of $\calC$ that contain $\sigma$ but do not contain the parent cell of $\sigma$. The set $\calC(\sigma)$ can be computed as follows. Let $\sigma'$ be the parent of $\sigma$ in $\Xi_{i-1}$. For each line in the conflict list $H_{\sigma'}$, it is the supporting line of a bounding ray of a cone $C$; we add $C$ to $\calC(\sigma)$ if $C$ contains $\sigma$.
Since $\sigma'$ has $O(1)$ children, computing $\calC(\sigma)$ for all children $\sigma$ of $\sigma'$ takes $O(|H_{\sigma'}|)$ time. As such, computing $\calC(\sigma)$ for all cells $\sigma\in \Xi$ can be done in time linear in the sum of the conflict sizes of all cells of $\Xi$, which is $O(nr)$. This also implies that $\sum_{\sigma\in \Xi}|\calC(\sigma)|=O(nr)$. 

For each cell $\sigma\in\Xi$, we further construct a binary search tree to store all cones of $\calC(\sigma)$ by the index order of their generators. Since we already know the sorted order of the generators of all cones of $\calC$, the sorted lists for $\calC(\sigma)$ for all cells $\sigma\in \Xi$ can be obtained in $O(nr)$ time. As such, constructing the binary search trees for all cells of $\Xi$ takes $O(nr)$ time in total. 

In addition, for each cell $\sigma$ in the last cutting $\Xi_t$, let $\calC_{\sigma}$ denote the set of the cones of $\calC$ whose bounding lines are in the conflict list $H_{\sigma}$. We construct a data structure of Lemma~\ref{lem:linearspace} on the cones of $\calC_{\sigma}$, which takes $O(n/r)$ space and $O(n/r\cdot \log(n/r))$ time since $|\calC_{\sigma}|=O(n/r)$. Let $\calD_{\sigma}$ denote the data structure. As $\Xi_t$ has $O(r^2)$ cells, constructing the data structures $\calD_{\sigma}$ for all cells $\sigma\in \Xi$ takes $O(nr)$ space and $O(nr\log (n/r))$ time. 

Finally, we construct the GH data structure for $P$ in $O(n)$ time and space~\cite{ref:GuibasOp89}. This finishes the preprocessing, which takes $O(nr)$ space and $O(nr\log (n/r))$ time in total. 

\subparagraph{Queries.}
Given a query point $q\in P_r$, we compute $s=\vis(d,q)$ in $O(\log n)$ time by the GH data structure. Next,
we will first compute $A_s(q)$ and then sort $A_s(q)$, after which $\vis_s(q)=\vis(P_\ell,q)$ can be obtained. 

We compute $A_s(q)$ as follows. Starting from the only cell of $\Xi_0$, we follow the hierarchical cutting in a top-down manner. For each $1\leq i\leq t$, assume that we already know the cell $\sigma'$ of $\Xi_{i-1}$ that contains $q$ (which is true initially when $i=1$ since the only cell of $\Xi_0$ is the entire plane). By checking all $O(1)$ children of $\sigma'$, we find the cell $\sigma\in \Xi_i$ containing $q$. Then, using the cone $C_q(s)$, we perform a clipping operation on the binary search tree of $\calC(\sigma)$ to find the generators of the cones of $\calC(\sigma)$ that are in $A_s(q)$, which takes $O(\log n)$ time. We do this until the last cutting $\Xi_t$. As $t=O(\log r)=O(\log n)$, the above takes $O(\log^2 n)$ time plus linear time in the output size. 
In addition, for the last cutting $\Xi_t$, suppose that $\sigma$ is the cell containing $q$. We find the generators of the cones of $\calC_{\sigma}$ that are in $A_s(q)$ using the data structure $\calD_{\sigma}$. By Lemma~\ref{lem:linearspace}, this takes $O((n/r)^{1/2+\epsilon})$ time plus linear time in the output size. This computes $A_s(q)$ in a total of $O((n/r)^{1/2+\epsilon}+\log^2 n+|\vis_s(q)|)$ time. The correctness is based on the following observation: If a generator of a cone $C\in \calC$ is in $A_s(q)$ and $\sigma$ is the cell of $\Xi_t$ containing $q$, then $C$ is either in $\calC_{\sigma}$ or in $\calC(\sigma')$ for an ancestor cell $\sigma'$ of $\sigma$ (we consider $\sigma$ an ancestor cell of itself). 

Next, sorting the elements of $A_s(q)$ can be done in $O(n^{\epsilon'}+k')$ time for any $\epsilon'>0$ by Lemma~\ref{lem:sort}. Hence, computing $\vis_s(q)$ takes $O(n^{\epsilon'}+(n/r)^{1/2+\epsilon}+\log^2 n+|\vis_s(q)|)$ time in total. Since $r\leq n^{1-\epsilon}$, the time is bounded by $O((n/r)^{1/2+\epsilon}+|\vis_s(q)|)$ if we choose $\epsilon'$ small enough.
\end{proof}

Finally, plugging Lemma~\ref{lem:tradeoff} into our framework leads to the following result. 

\begin{theorem}\label{theo:tradeoff}
Given a simple polygon $P$ with $n$ vertices, let $r$ be a parameter satisfying $1 \le r \le n^{1-\epsilon}$ for any constant $\epsilon > 0$. A data structure of size $O(nr + n\log (n/r))$ can be constructed in $O((nr + n\log (n/r))\log n)$ time such that $\vis(q)$ can be computed in $O((n/r)^{1/2+\epsilon})$ time for any query point $q \in P$.  
In particular, setting $r = 1$ yields a data structure of $O(n\log n)$ space, $O(n\log^2 n)$ preprocessing time, and $O(n^{1/2+\epsilon})$ query time.
\end{theorem}
\begin{proof}
We first construct a balanced decomposition of $P$ together with the BD-tree $\calT(P)$. Let $t=n/r$. 

For each internal node $v$ of $\calT(P)$, the diagonal $d_v$ divides $P_v$ into two subpolygons. We construct a data structure of Lemma~\ref{lem:tradeoff} for each subpolygon with respect to $d_v$
with parameter $r_v=\max\{1,n_v/t\}$ (as the $r$ in Lemma~\ref{lem:tradeoff}), 
where $n_v=|P_v|$; let $\calD_v$ denote these two data structures. 
If $r_v=1$, we call $v$ a {\em lower node}; otherwise $v$ is an {\em upper node}. Note that if $v$ is an upper node, then $n_v/r_v=n/r$. 

In addition, we construct a point location data structure on the triangulation formed by the leaf triangles of $\calT(P)$. This finishes our preprocessing. We next analyze the space and the preprocessing time. 

Computing the balanced decomposition $\calT(P)$ can be done in $O(n)$ time~\cite{ref:GuibasOp89}. Constructing the point location data structure takes $O(n)$ space and  time~\cite{ref:KirkpatrickOp83,ref:EdelsbrunnerOp86}. In the following, we analyze the data structures $\calD_v$ of the nodes $v$ of $\calT(P)$. 

First of all, if $r_v=1$ for a node $v$, then by definition, $n_v/t\leq 1$ and thus $n_v\leq t$. As such, all descedants $v'$ of $v$ are lower nodes because $n_{v'}\leq n_v$. This implies that all upper nodes form a connected subtree of $\calT(P)$ containing the root. Also, since $n/r=n_v/r_v$ for upper nodes $v$ and $n_v$ is geometrically decreasing along the path from the root to any leaf of $\calT(P)$, the depth of every lower node is $\Omega(\log r)$. As the height of $\calT(P)$ is $O(\log n)$, all lower nodes are in the lowest $\Theta(\log n-\log r)=\Theta(\log (n/r))$ levels of $\calT(P)$. Since the space of $\calD_v$ for each lower node $v$ is $O(n_v)$, the total space of $\calD_v$ of the lower nodes $v$ in each level of $\calT(P)$ is $O(n)$. Hence, the total space of $\calD_v$ of all lower nodes $v$ of $\calT(P)$ is $O(n\log(n/r))$. 

For the upper nodes, let $\calT'$ denote the subtree comprising all upper nodes. 
By the definition of the balanced polygon decomposition of $P$, we have $n/3\leq n_v\leq 2n/3$ for each child $v$ of the root. Suppose $S(n)$ is the total space of $\calD_v$ of all nodes $v$ of $\calT'$. Note that the space of $\calD_v$ for the root $v$ is $O(nr)=O(n^2/t)$. Then, we have the following recurrence relation:
\[S(n)=\max_{n/3\leq m \leq 2n/3}\{S(m)+S(n-m+2)\}+O(n^2/t).\]
The recurrence solves to $S(n)=O(n^2/t)$, which is $O(nr)$. Hence, the total space of the data structures $\calD_v$ for all upper nodes $v$ is $O(nr)$. 

In summary, the total space of the data structures $\calD_v$ for all nodes $v$ of $\calT(P)$ is $O(nr+n\log(n/r))$. The preprocessing time analysis is similar, but has an additional $\log n$ factor. 

Given a query point $q$, we follow the query algorithm framework described in the beginning of Section~\ref{sec:bdscheme} to process the nodes of $\calT(P)$ along the path from the leaf whose triangle containing $q$ to the root. For each node $v$, we apply the query algorithm of the data structure $\calD_v$ with query time $O((n_v/r_v)^{1/2+\epsilon})$ plus the output size. If $v$ is an upper node, then $(n_v/r_v)^{1/2+\epsilon}=O((n/r)^{1/2+\epsilon})$ since $n_v/r_v=n/r$. If $v$ is a lower node, then $r_v=1$ and $n_v\leq t=n/r$. Hence, $(n_v/r_v)^{1/2+\epsilon}=(n_v)^{1/2+\epsilon}=O((n/r)^{1/2+\epsilon})$. As such, in addition to the output size, we spend $O((n/r)^{1/2+\epsilon})$ time at each node $v$. Therefore, the total query time for computing $\vis(q)$ is $O((n/r)^{1/2+\epsilon}\log n+|\vis(q)|)$. As $r\leq n^{1-\epsilon}$, we can write the time bound as $O((n/r)^{1/2+\epsilon}+|\vis_s(q)|)$ since the $\log n$ factor is absorbed by $(n/r)^{\epsilon}$.
\end{proof}



\section{A new polygon decomposition}
\label{sec:decomp}


In this section, we introduce a new decomposition of $P$ together with its associated data structure for answering visibility queries.

Our decomposition decomposes $P$ into {\em geodesic triangles}, a concept first introduced in \cite{ref:ChazelleRa94}, in which a balanced geodesic triangulation of $P$ was proposed and it also decomposes $P$ into geodesic triangles. However, our decomposition differs fundamentally from the balanced geodesic triangulation in~\cite{ref:ChazelleRa94} and possesses several key properties, such as the gate property described in Section~\ref{sec:overview}, that are essential for the efficiency of our visibility query algorithm.  

In the following, after defining geodesic triangles, we present a data structure for handling visibility queries with respect to a single geodesic triangle in Section~\ref{sec:twosubproblems}. We then describe our new polygon decomposition in Section~\ref{sec:newdecom}.  


\subparagraph{Geodesic triangles.}
The {\em geodesic path} $\pi(p,q)$ is a shortest path in $P$ between two points $p$ and $q$. 
For three vertices $v_i$, $v_j$, and $v_k$ of $P$ ordered clockwise along $\partial P$, consider the three geodesic paths $\pi(v_i,v_j)$, $\pi(v_j,v_k)$, and $\pi(v_k,v_i)$. If we remove the subpaths shared by any two of these three paths, then the region bounded by the remaining portions of the three paths is a {\em geodesic triangle}; see Figure~\ref{fig:geotriangle}.
A geodesic triangle has three {\em sides}, each being a concave chain that is a subpath of one of the three geodesic paths, and three {\em apexes}, corresponding to the endpoints of its sides.  
In general, a geodesic triangle has exactly three apexes, although it may degenerate to an empty region when one of the three vertices $v_i$, $v_j$, or $v_k$ lies on the geodesic path connecting the other two.

\begin{figure}[t]
\centering
\includegraphics[height=2.0in]{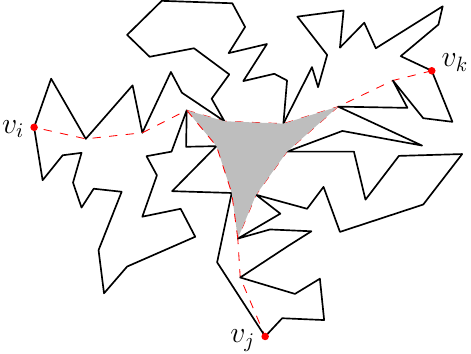}
\caption{Illustrating a geodesic triangle (the gray region) formed by three geodesic paths.}
\label{fig:geotriangle}
\end{figure}

\subsection{A data structure for geodesic triangles}
\label{sec:twosubproblems}


Consider a geodesic triangle $\triangle$, defined by three vertices $p_i$, $p_j$, and $p_k$ as discussed above.  
Each edge of $\triangle$ is either an edge of $P$ or a diagonal.  
For each diagonal $d$ on $\partial \triangle$, let $P_d(\triangle)$ denote the subpolygon of $P$ bounded by $d$ that does not contain $\triangle$.  
We refer to $P_d(\triangle)$ as the {\em pocket of $P$ separated by $d$ from $\triangle$}.  
For simplicity, if $d$ is an edge of $P$, we let $P_d(\triangle)$ denote $d$ itself.  
When $\triangle$ is clear from context, we simply write $P_d$.  

We study two subproblems with respect to $\triangle$.  
In the first subproblem, the query point $q$ lies in $\triangle$, and the goal is to compute $\vis(q)$.  
In the second subproblem, the query point $q$ lies in $P_d$ for some diagonal $d$ of $\partial\triangle$, and we wish to compute $\vis(P \setminus P_d, q)$, i.e., the portion of $\vis(q)$ that lies outside the pocket $P_d$.  
Our result is summarized in the following lemma, which will be used later.  

\begin{lemma}\label{lem:geotriangle}
With respect to a geodesic triangle $\triangle$, a data structure of size $O(n^2)$ can be constructed in $O(n^2 \log n)$ time such that, for any query point $q$,  
if $q \in \triangle$, then $\vis(q)$ can be computed in $O(\log n)$ time,  
and if $q \in P_d$ for some diagonal $d$ on the boundary of $\triangle$, then $\vis(P \setminus P_d, q)$ can be computed in $O(\log n)$ time.  
\end{lemma}

Intuitively, the motivation for having a query algorithm that computes $\vis(P \setminus P_d, q)$ is as follows.  
If $\vis(P_d)$ is already available, then after obtaining $\vis(P \setminus P_d, q)$, we can compute $\vis(q)$ by merging the two.  
To compute $\vis(P_d)$ itself, we would like to apply the same algorithm recursively.  
For this recursive approach to work, we require an ``appropriate'' polygon decomposition of $P$.  
This motivates the new decomposition introduced in Section~\ref{sec:newdecom}.

In the rest of this subsection, we prove Lemma~\ref{lem:geotriangle}. 

\subsubsection{The first subproblem of Lemma~\ref{lem:geotriangle}: queries from a geodesic triangle}
\label{sec:firstsub}

In the first subproblem, we wish to compute $\vis(q)$ for any query point $q\in \triangle$.

Let $\pi$ be a side of $\triangle$, which is a concave geodesic path. 
Define $P_{\pi}$ to be the union of $P_d$ for all edges $d\in \pi$. Note that $\triangle$ and $P_{\pi}$ for all three sides $\pi$ of $\triangle$ together form a partition of $P$. 
In the following, we show how to compute $\vis(P_{\pi},q)$ 
for a side $\pi$ of $\triangle$. $\vis(s)$ can then be obtained by merging $\vis(P_{\pi},q)$ for all three sides $\pi$ of $\triangle$. 

If $\pi$ consists of a single diagonal $d$, then we can use the quadratic-space method described in Section~\ref{sec:diasep}. We follow the notation there, e.g., $R_d$, $\calA_d$, $\vis_d(q)$. Let $n_d=|P_d|$. Note that the arrangement $\calA_d$ has $O(n_d^2)$ cells. Hence, we can construct a data structure of $O(n_d^2)$ space in $O(n_d^2\log n_d)$ time such that $\vis_d(q)$ can be computed in $O(\log n)$ time.

If $\pi$ has more than one diagonal, we extend the idea to compute $\vis_{\pi}(q)$, defined as the union of $\vis_d(q)$ for all edges $d\in \pi$. Hence, $\vis_{\pi}(q)$ is the {\em exterior visibility polygon} of $q$ in $P_{\pi}$ through $\pi$. In the preprocessing, we overlay the arrangements $\calA_d$ of all edges $d\in \pi$ to obtain an arrangement $\calA_{\pi}$. Then, as above for $\calA_d$, points in the same cell of $\calA_{\pi}$ have the same $\vis_{\pi}(q)$, and $\vis_{\pi}(q)$ in two adjacent cells in $\calA_{\pi}$ differs by at most one element of $P_{\pi}$. The size of $\calA_{\pi}$ is $O(n^2_{\pi})$ with $n_{\pi}=|P_{\pi}|$. As before, we can use persistent binary search trees~\cite{ref:SnarnakPl86,ref:DriscollMa89} to store $\vis_{\pi}(q)$ for all cells of $\calA_{\pi}$ in a total of $O(n_{\pi}^2)$ space and $O(n_{\pi}^2\log n_{\pi})$ time. We also construct a point location data structure on $\calA_{\pi}$ in $O(n_{\pi}^2)$ time and space. 
Given a query point $q$, in $O(\log n_d)$ time we can locate the cell of $\calA_{\pi}$ containing $q$ and thus obtain the access to a binary search tree storing $\vis_{\pi}(q)$. 

After having $\vis_{\pi}(q)$, we can compute $\vis(P_{\pi},q)$ as follows. 
We determine the portion of $\pi$ visible to $q$ in $P$. Notice that since $q\in \triangle$ and the three sides of $\triangle$ are concave, the portion of $\pi$ visible to $q$ is a contiguous subpath of $\pi$ delimited by a cone $C_{\pi}(q)$ whose bounding lines can be determined by computing the tangents from $q$ to the other two sides of $\triangle$ in $O(\log n)$ time. After $C_{\pi}(q)$ is computed, we perform a clipping operation on the tree storing $\vis_{\pi}(q)$, after which we obtain a binary search tree storing $\vis(P_{\pi},q)$. 

Finally, we can compute $\vis(q)$ as follows. 
We first compute a binary search tree storing $\vis(P_{\pi},q)$ for each side $\pi$ of $\triangle$, which takes $O(\log n)$ time as discussed above. Since each side $\pi$ of $\triangle$ is concave, the cones $C_q(\pi)$ for the three sides $\pi$ of $\triangle$ together partition the full angular range around $q$. This means that by simply merging the three trees for $\vis(P_{\pi},q)$ for the three sides $\pi$ of $\triangle$, we can obtain a binary search tree storing $\vis(s)$. This merge operation of binary search trees can be done in $O(\log n)$ time. 


In summary, with $O(n^2)$ space and $O(n^2\log n)$ time preprocessing, $\vis(q)$ can be computed in $O(\log n)$ time for any query point $q\in \triangle$. This proves the first part of Lemma~\ref{lem:geotriangle}.

\subsubsection{The second subproblem of Lemma~\ref{lem:geotriangle}: queries from a pocket}
\label{sec:secondsub}

In the second subproblem, the query point $q$ is in $P_d$ for some diagonal $d$ on $\partial \triangle$ and we wish to compute $\vis(P\setminus P_{d},s)$. We can slightly modify the above algorithm to solve this subproblem as follows. 

Let $\pi$ be the side of $\triangle$ that contains $d$. Since $\pi$ is concave, it is not difficult to see that $q$ cannot be visible to any point in the interior of $P_{d'}$ for any other diagonal $d'$ in $\pi$. Hence, $\vis(P\setminus P_d,q)=\vis(P\setminus P_{\pi},q)$. In the following we focus on computing  $\vis(P\setminus P_{\pi},q)$.

Let $\pi_1$ and $\pi_2$ be the other two sides of $\triangle$. Note that $P\setminus P_{\pi}=\triangle\cup P_{\pi_1}\cup P_{\pi_2}$. 
We construct the same data structure as for the first subproblem above for the sides $\pi_1$ and $\pi_2$. In addition, we construct the GH data structure~\cite{ref:GuibasOp89} in $O(n)$ time. This finishes our preprocessing, which takes $O(n^2)$ space and $O(n^2\log n)$ time.

\begin{figure}[t]
\centering
\includegraphics[height=2.2in]{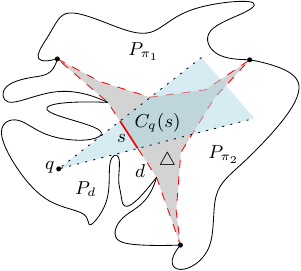}
\caption{The visibility of $q$ in $P\setminus P_{\pi}$ is constrained by the cone $C_q(s)$.}
\label{fig:secondsub}
\end{figure}

Given a query point $q\in P_d$, we first find the segment $s$ of $d$ that is visible to $q$, which can be done in $O(\log n)$ time using the GH data structure~\cite{ref:GuibasOp89}. If $s=\emptyset$, then $\vis(P\setminus P_{\pi},q)=\emptyset$. Otherwise, the visibility of $q$ in $P\setminus P_{\pi}$ is constrained by the cone $C_q(s)$ defined by $q$ and $s$; see Figure~\ref{fig:secondsub}.
As above, we first compute a binary search tree storing $\vis_{\pi_1}(q)$. The portion of $\pi_1$ visible to $q$ is determined by a cone $C_q(\pi_1)\subseteq C_q(s)$. The bounding lines of $C_q(\pi_1)$ can be determined by first computing the tangents from $q$ to $\pi_1$ and $\pi_2$ in $O(\log n)$ time. We then do a clipping operation on the tree storing $\vis_{\pi_1}(q)$ by using $C_q(\pi_1)$ and thus obtain a tree for $\vis(P_{\pi_1},q)$, which takes $O(\log n)$ time. Similarly, we can obtain a tree storing $\vis(P_{\pi_2},q)$ in $O(\log n)$ time. Note that the two cones $C_q(\pi_1)$ and $C_q(\pi_2)$ are interior disjoint and share a common bounding line. Hence, by merging the two trees, we obtain a binary search tree storing the combinatorial representation of $\vis(P_{\pi_1}\cup P_{\pi_2},q)$, which is also the combinatorial representation of $\vis(\triangle\cup P_{\pi_1}\cup P_{\pi_2},q)=\vis(P\setminus P_{\pi},q)=\vis(P\setminus P_{d},q)$, because all three sides of $\triangle$ are concave. 

In summary, we can compute $\vis(P\setminus P_d,q)$ in $O(\log n)$ time for a query point $q\in P_d$. The preprocessing takes $O(n^2)$ space and $O(n^2\log n)$ time. This proves the second part of Lemma~\ref{lem:geotriangle}.

\subsection{The new decomposition}
\label{sec:newdecom}

We now introduce a new decomposition of $P$, defined with respect to a designated edge of $P$.  
Let $p_1, p_2, \ldots, p_n$ be the vertices of $P$ ordered clockwise along $\partial P$.  
In the following, we describe the decomposition with respect to the edge $\overline{p_1p_n}$.


We represent the decomposition by a {\em decomposition tree} $\Psi$.  
Let $v^*$ denote the root of $\Psi$.  
Each node $v$ of $\Psi$ corresponds to a subpolygon $P_v$ and a {\em gate} $d_v$, where $d_v$ is a diagonal of $P$ that separates $P_v$ from the ``outside world'' $P \setminus P_v$, and all other edges of $P_v$ are edges of $P$.  
Initially, we have $P_{v^*} = P$ and $d_{v^*} = \overline{p_1p_n}$  (since $v^*$ is the root, as a special case we allow $d_{v^*}$ not to be a diagonal, i.e., we assume $d_{v^*}$ connects $P$ with the outside).  

Each node $v$ also corresponds to a subpolygon $\core(v) \subseteq P_v$, called the {\em core} of $P_v$, which is the union of a collection of geodesic triangles. The core $\core(v)$ and the subpolygons $P_u$ of all children $u$ of $v$ form a partition of $P_v$ (thus $P_u \subseteq P_v$ for each child $u$).  
Moreover, for every child $u$ of $v$, we have $|P_u| \le |P_v|/2$.  
Hence, the sizes of the subpolygons $P_v$ decrease geometrically along any root-to-leaf path in $\Psi$, and the height of $\Psi$ is $O(\log n)$.  

At each node $v$ of $\Psi$, we explicitly store the edges of $\core(v)$.  
The cores of all nodes are interior-disjoint and together form a partition of $P$.  
Since no new vertices are introduced and no two edges of the decomposition intersect, the overall size of the decomposition and thus the total size of all cores is $O(n)$.  
Because each core consists of geodesic triangles, our decomposition effectively partitions $P$ into geodesic triangles.  

As will be seen later, when our decomposition is used to answer visibility queries, the property that the gate $d_v$ separates $P_v$ from $P \setminus P_v$ plays a crucial role.  
We refer to this as the {\em gate property}.  
Moreover, for any node $v$ and any ancestor $u$ of $v$, since $P_v \subseteq P_u \subseteq P$, the gate property also implies that $d_v$ separates $P_v$ from $P_u \setminus P_v$.




Section~\ref{sec:decompdef} presents the detailed construction of the decomposition $\Psi$. 
The following Theorem~\ref{theo:treequeries} is our main result for answering visibility queries using $\Psi$. 
Our results in Sections~\ref{sec:Optimal_query_time} and \ref{sec:A_Quadratic_Space_Structure} will build upon this theorem.  

\begin{theorem}\label{theo:treequeries}
A data structure of size $O(n^2 \log^4 n)$ can be constructed in $O(n^2 \log^5 n)$ time for the decomposition $\Psi$ such that, for any query point $q$,  
if $q$ lies in the core $\core(u)$ of a node $u \in \Psi$, then $\vis(P_u, q)$ can be computed in $O(\log n)$ time;  
and if $q$ lies in $P_v$ for a child $v$ of $u$, then $\vis(P_u \setminus P_v, q)$ can be computed in $O(\log n)$ time.
\end{theorem}
\begin{proof}
We only sketch the proof here; full details appear in Section~\ref{sec:theo:treequeries}.  

Consider the root $v^*$ of $\Psi$.  
For each geodesic triangle $\triangle$ in the core $\core(v^*)$, we build the data structure of Lemma~\ref{lem:geotriangle} for $\triangle$ in $P_{v^*}$ (which equals $P$).  
We can show that this construction requires $O(n^2 \log^4 n)$ space and $O(n^2 \log^5 n)$ preprocessing time.  
Afterward, for any query point $q \in P$, if $q \in \core(v^*)$, we can compute $\vis(P_{v^*}, q)$ in $O(\log n)$ time;  
and if $q \in P_v$ for a child $v$ of $v^*$, we can compute $\vis(P_{v^*} \setminus P_v, q)$ in $O(\log n)$ time.  

In general, for each node $v \in \Psi$, and for each geodesic triangle $\triangle$ in $\core(v)$, we build the data structure of Lemma~\ref{lem:geotriangle} for $\triangle$ within $P_v$.  
This requires $O(n_v^2 \log^4 n_v)$ space and $O(n_v^2 \log^5 n_v)$ preprocessing time, where $n_v = |P_v|$.  
Since the sizes of the subpolygons $P_v$ decrease geometrically along every root-to-leaf path, and since the subpolygons $P_v$ of nodes at the same level of $\Psi$ are interior disjoint,  
the total space over all nodes of $\Psi$ is $O(n^2 \log^4 n)$, and the total preprocessing time is $O(n^2 \log^5 n)$.  
After this preprocessing, the queries can be answered as stated in the theorem.  
\end{proof}

As a quick demonstration on why $\Psi$ is useful, we show that $\vis(q)$ can be computed in $O(\log^2 n)$ time. 
Given a query point $q$, we first use a point location query to find the node $v$ whose core $\core(v)$ contains $q$. Then, we compute $\vis(P_v,q)$ in $O(\log n)$ time by Theorem~\ref{theo:treequeries}. Next, we compute $\vis(P_u\setminus P_v,q)$ in $O(\log n)$ time, where $u$ is the parent of $v$. Note that $\vis(P_v,q)\cup \vis(P_u\setminus P_v,q)$ is $\vis(P_u,q)$. Since $P_v$ is separated from $P_u\setminus P_v$ by the diagonal $d_v$, we can obtain $\vis(P_u,q)$ by merging $\vis(P_v,q)$ and $\vis(P_u\setminus P_v,q)$ in $O(\log n)$ time by Lemma~\ref{lem:merge}. If we continue this until the root, we will obtain $\vis(q)$. The total time is $O(\log^2 n)$ as the height of $\Psi$ is $O(\log n)$. 

\subparagraph{Remark.}
As in~\cite{ref:ChazelleRa94}, our decomposition partitions $P$ into 
geodesic triangles, although it is more intricate to describe. 
A natural question is whether the geodesic triangulation of~\cite{ref:ChazelleRa94} 
can be used instead. We initially pursued this approach, but several 
difficulties arose. In particular, one might hope that visibility always 
propagates through parent geodesic triangles; however, this is not the case. 
Roughly speaking, visibility proceeds through ``kites'' (following the 
terminology of~\cite{ref:ChazelleRa94}), rather than through the triangles 
themselves. In contrast, in our decomposition, as illustrated by the example 
above, visibility is always routed through the gates.
Moreover, it seems unlikely that one can obtain results comparable to 
Theorem~\ref{theo:treequeries} using the standard geodesic triangulation in \cite{ref:ChazelleRa94}. 
Indeed, such an approach would require preprocessing each geodesic triangle 
using Lemma~\ref{lem:geotriangle} (as we did in the above sketched proof of Theorem~\ref{theo:treequeries}). 
Since there are $O(n)$ triangles and each 
requires quadratic preprocessing time, the total preprocessing time would 
become cubic. To overcome these obstacles, a more careful redesign of the 
triangulation process is necessary, which leads to our decomposition (e.g., we can show that the total preprocessing time for all our geodesic triangles is only near-quadratic as stated in Theorem~\ref{theo:treequeries}). 
From this perspective, our approach may be viewed as a variation of the 
geodesic triangulation in~\cite{ref:ChazelleRa94} that is better suited 
for visibility queries.


\subsection{The construction of the decomposition $\boldsymbol{\Psi}$}
\label{sec:decompdef}
We now present the details of our decomposition and the decomposition tree $\Psi$. 

During the discussion, instead of directly arguing the gate property, we will establish two other properties that together guarantee the gate property as shown 
in the following lemma. 

\begin{lemma}\label{lem:gate}
Suppose that the following two properties hold for each non-root node $v\in \Psi$: (1) $d_v$ separates $P_v$ from $P_u\setminus P_v$, where $u$ is the parent of $v$; (2) the gate $d_{u}$ of $u$ is not in $P_v$. Then, the gate property holds for $v$, i.e., $d_v$ separates $P_v$ from $P\setminus P_v$.    
\end{lemma}
\begin{proof}
We prove the lemma by induction. 
Initially, if $v$ is a child of the root $v^*$, then since $d_v$ separates $P_v$ from $P_u\setminus P_v$, $u=v^*$, and $P_{v^*}=P$, 
it is vacuously true that $d_v$ separates $P_v$ from $P\setminus P_v$. 

We now inductively assume that $d_u$ separates $P_u$ from $P\setminus P_u$. We next argue that $d_v$ separates $P_v$ from $P\setminus P_v$. Indeed, consider two points $p$ and $p'$ with $p\in P_v$ and $p'\in P\setminus P_v$. It suffices to show that the geodesic path $\pi(p',p)$ must intersect $d_v$. If $p'$ is in $P_u$, then since $d_v$ separates $P_v$ from $P_u\setminus P_v$, then $\pi(p',p)$ must intersect $d_v$. If $p'\not\in P_u$, then $p'\in P\setminus P_u$.
In this case, since $p\in P_v\subseteq P_u$,  $d_u\not\subseteq P_v$, and $d_u$ separates $P_u$ from $P\setminus P_u$, if we traverse from $p'$ to $p$ along $\pi(p',p)$, we must first cross $d_u$ and then $d_v$. The lemma thus follows. 
\end{proof}

We next describe our decomposition. 
We will focus on explaining the children $v$ of $v^*$. 
Specifically, we will define the core $\core(v^*)$, the subpolygons $P_v$, and the gate $d_v$ for every child $v$ of $v^*$. In particular, we will argue that the following {\em three key properties} hold for each $v$: (1) $P_v$ is separated from $P\setminus P_v$ by $d_v$; (2) $d_{v^*}$ is not in $P_v$; 
(3) $P_v$ has no more than $n/2$ edges other than the gate $d_v$. Note that the first two properties are those required in Lemma~\ref{lem:gate}. 
To define the whole tree $\Psi$, we simply apply the same decomposition recursively on $P_v$ for each child $v$ of $v^*$.

Let $p^*=p_{n/2}$, the middle point of $\partial P$ (recall that vertices of $P$ are indexed from $p_1$ to $p_n$ clockwise around $\partial P$). We consider the geodesic triangle $\triangle^*$ formed by the three geodesic paths $\pi(p^*,p_1)$, $\pi(p^*,p_n)$, and $\pi(p_1,p_n)$. Note that $\pi(p_1,p_n)$ is the edge $d_{v^*}=\overline{p_1p_n}$. We first consider the general case where $\triangle^*$ is not empty, and we will discuss later that the empty case can be treated similarly. We add $\triangle^*$ to the core $\core(v^*)$. 

Note that $d_{v^*}$ is one side of $\triangle^*$ and thus $p_1$ and $p_n$ are two apexes of $\triangle^*$. Let $b^*$ be the third apex (see Figure~\ref{fig:newdecomp}). Note that $\partial \triangle^* \setminus d_{v^*}$ is the union of the other two sides of $\triangle^*$. 

\begin{figure}[t]
\centering
\includegraphics[height=1.5in]{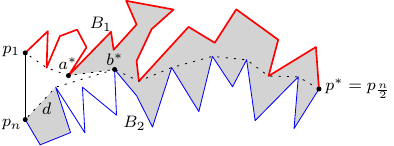}
\caption{The shaded region bounded by $d$ is a pocket $P_d(\triangle^*)$. $\overline{a^*b^*}$ is the bride $d^*$ on $\partial\triangle^*\setminus d_{v^*}$. The bigger shaded region bounded by $d^*$ is $P_{d^*}(\triangle^*)$.}
\label{fig:newdecomp}
\end{figure}

Let $B_1$ denote the boundary of $P$ from $p_1$ clockwise to $p^*$ and $B_2$ the boundary of $P$ from $p^*$ clockwise to $p_n$. We let both $B_1$ and $B_2$ contain $p^*$.  Hence, $\partial P=B_1\cup B_2\cup d_{v^*}$. We say that 
a diagonal is a {\em bridge} if one vertex of it is in $B_1\setminus\{p^*\}$ while the other is in $B_2\setminus\{p^*\}$. By definition, if $b^*=p^*$, then $\partial \triangle^* \setminus d_{v^*}$ has no bridge; otherwise, $\partial \triangle^* \setminus d_{v^*}$ has exactly one bridge, denoted by $d^*$, which must have $b^*$ as a vertex (see Figure~\ref{fig:newdecomp}, where $d^*=\overline{a^*b^*}$). 

For each diagonal $d$ of $\partial \triangle^* \setminus d_{v^*}$, if $d$ is not a bridge, then it connects two vertices either both on $B_1$ or both on $B_2$. Recall we defined in Section~\ref{sec:twosubproblems} that the notation $P_d(\triangle^*)$ refer to the pocket of $P$ separated by $d$ from  $\triangle^*$ (see Figure~\ref{fig:newdecomp}). 
Observe that the edges of $P_d(\triangle^*)$ other than $d$ are all in $B_1$ or all in $B_2$. Hence, $P_d(\triangle^*)$ has no more than $n/2$ edges other than $d$. We create a child $v$ for $v^*$ in $\Psi$ with $P_v=P_d(\triangle^*)$ and the gate $d_v=d$. In this way, $d_v$ separates $P_v$ from $P\setminus P_v$ and $d_{v^*}=\overline{p_1p_n}$ is not an edge of $P_v$. Hence, all three key properties hold for $v$. 

If $d$ is a bridge, i.e., $d=d^*$, then $b^*\neq p^*$. In this case, $P_{d^*}(\triangle^*)$ contains vertices from both $B_1$ and $B_2$ (see Figure~\ref{fig:newdecomp}), and therefore, $P_{d^*}(\triangle^*)$ may have more than $n/2$ edges other than $d^*$. We therefore cannot simply create a child for $v^*$ as above and need to resort to other approaches as follows. 


\subparagraph{Decomposing $\boldsymbol{P_{d^*}(\triangle^*)}$ into subpockets.}
Let $a^*$ be the vertex of $d^*$ other than $b^*$ (see Figure~\ref{fig:newdecomp}). Then, $\pi(p^*,a^*)$ is a subpath of either $\pi(p^*,p_1)$ or $\pi(p^*,p_n)$. 
To simplify the notation, let $\calQ=P_{d^*}(\triangle^*)$ and $\pi^*=\pi(a^*,p^*)$.

The geodesic path $\pi^*$ decomposes the pocket $\calQ$ into {\em subpockets} each of which is bounded by a subpath of $\pi^*$ and a sequence of edges of $B_1$ or $B_2$. If a subpocket is bounded by a subpath of $\pi^*$ and a sequence of edges of $B_1$ (resp., $B_2$), we call it a {\em  $B_1$-bounded subpocket} (resp., {\em $B_2$-bounded subpocket}); see Figure~\ref{fig:subpockets}. For each subpocket $Q$, the subpath of $\pi^*$ on the boundary of $Q$ is called the {\em base} of $Q$. 
We say that $Q$ is {\em closed} if its base is a single edge; otherwise, it is {\em open} (see Figure~\ref{fig:subpockets}). 

For a closed subpocket $Q$, since its edges other than its base edge are either all from $B_1$ or all from $B_2$, the number of its edges other than its base is at most $n/2$, and $Q$ is separated from $P\setminus Q$ by its base edge. Hence, 
we create a child $v$ for $v^*$ in $\Psi$ with $P_v=Q$ and $d_v$ as the base edge of $Q$. Clearly, $d_{v^*}$ is not in $P_v$. Hence, all three key properties hold for $v$.

\begin{figure}[t]
\centering
\includegraphics[height=1.5in]{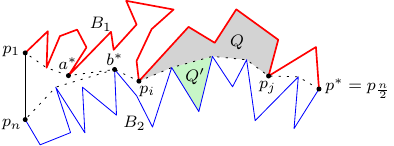}
\caption{The gray region $Q$ is an open $B_1$-bounded subpocket while the green region $Q'$ is a closed $B_2$-bounded subpocket.}
\label{fig:subpockets}
\end{figure}

If $Q$ is open, then we will need to further decompose it. The situation in this case becomes significantly more complicated and our main effort is to handle this case. Without loss of generality, we assume that $Q$ is $B_1$-bounded. 

\subparagraph{Decomposing an open subpocket $\boldsymbol{Q}$.}
Define $Q^*=P\setminus \calQ$. We consider $Q^*$ a {\em special subpocket}. Following our above definition, $Q^*$ is a closed subpocket since it has only one edge of $\pi^*$. To differentiate, when we say ``subpockets of $P$'', we refer to all subpockets of $\calQ$ and also $Q^*$; when we say ``subpockets of $\calQ$'', we only refer to the subpockets inside $\calQ$. Unless otherwise stated, a subpocket of $P$ refers to any subpocket of $\calQ\cup \{Q^*\}$. 

To simplify the discussion, we assume that each edge $e$ of $\pi^*$ is a diagonal. Indeed, this is the most general case because if $e$ is an edge of $P$ then we can assume there is an infinitely small subpocket separated by $e$. With this assumption, $e$ is on the boundaries of two subpockets of $P$. In particular, if $p^*$ is a vertex of $e$, then one subpocket bounded by $e$ is closed while the other is open. In this case, we consider $e$ a bridge in the open subpocket but not a bridge in the closed one. In this way, we have the following observation. 
\begin{observation}
Each open subpocket $Q$ of $\calQ$ has exactly two bridges, which are the first and last edges of its base (see Figure~\ref{fig:subpockets}).    
\end{observation}
\begin{proof}
To see this, since the base of $Q$ is a geodesic path (because it is a subpath of $\pi^*$) and $Q$ is a simple polygon, the base of $Q$ must be a concave chain, meaning that if we traverse on it, either we always make right turns or we always make left turns. Hence, the two end vertices of the base of $Q$ must be in $B_1$ while all other vertices are in $B_2$ because $Q$ is $B_1$-bounded (see Figure~\ref{fig:subpockets}). This implies that among all edges on the base of $Q$, only the first and last edges are bridges. For other edges of $Q$ not on its base, by definition, each of them connects two vertices of $B_1$, meaning that it is not a bridge. 
\end{proof}

We say that two subpockets of $P$ are {\em neighboring} if they share a common edge of $\pi^*$. By definition, for two neighboring subpockets of $\calQ$, one of them is $B_1$-bounded while the other is $B_2$-bounded. Also, a closed subpocket has exactly one neighboring subpocket while an open one has multiple. For any open subpocket $Q$, since only the first and last edges of its base are bridges, $Q$ has at most two neighboring open subpockets. The following simple observation implies that to compute $\vis(q)$ for a point $q\in Q$,  it suffices to consider $Q$ and its neighboring subpockets. 

\begin{observation}\label{obser:vissubpocket}
For any point $q\in Q$, if a point $p\in P$ is visible to $q$, then $p$ is either in $Q$ or in a neighboring subpocket of $Q$.
\end{observation}
\begin{proof}
Assume to the contrary $p$ is in other subpockets. Then, $\overline{pq}$ must cross two edges of $\pi^*$. Since $\pi^*$ is a geodesic path in $P$, no line segment in $P$ can cross two edges of $\pi^*$. We thus obtain contradiction. 
\end{proof}

With the above concepts, we next discuss how to further decompose an open subpocket $Q$ (again assume that $Q$ is $B_1$-bounded). 

Let $\pi(Q)$ denote the base of $Q$. Let $p_i$ be the end vertex of $\pi(Q)$ that is closer to $a^*$ along $\pi^*=\pi(a^*,p^*)$ and $p_j$ the other end vertex (see Figure~\ref{fig:subpockets}). By definition, $i<j$. Note that if we traverse $\pi(Q)$ from $p_i$ to $p_j$, we always make right turns. For simplicity of discussion, we consider the order of vertices of $\pi(Q)$ from $p_i$ to $p_j$ as the left-to-right order so that we can use ``left'' and ``right'' to refer to the directions of $\pi(Q)$. Note that $\pi(Q)$ could be spiral. Let $B_1(Q)$ denote the portion of $B_1$ in $Q$. Hence, $p_i$ is the first vertex and $p_j$ is the last vertex of $B_1(Q)$ if we order vertices of $B_1(Q)$ along $B_1$ clockwise (and we also consider it the left-to-right order of $B_1(Q)$). 
Also, the union of $B_1(Q)$ and $\pi(Q)$ forms the boundary $\partial Q$ of $Q$.

We pick the middle vertex $p_{k}$ of $B_1(Q)$ with $k=(i+j)/2$, and connect $p_k$ to $p_i$ and $p_j$ by geodesic paths, which yields a geodesic triangle, denoted by $\triangle_k$. Since the base $\pi(Q)$ is a geodesic path and $p_k$ is not on $\pi$, $\triangle_k$ cannot be empty. One side of $\triangle_k$ is on $\pi(Q)$ and we call it the {\em base side}; the other two sides are called the {\em non-base sides}. Among the two non-base sides, the one that connects the leftmost vertex of the base side is called the {\em left side} while the other is called the {\em right side} (recall that we have a left-to-right order of $\pi(Q)$; since the base side is a subpath of $\pi(Q)$, we can talk about the ``leftmost'' and ``rightmost'' vertices). We add $\triangle_k$ to the core $\core(v^*)$ of the root $v^*$ of $\Psi$.

\begin{figure}[t]
\centering
\includegraphics[height=2.8in]{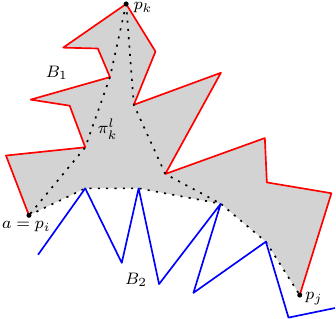}
\caption{Illustrating the case $a=p_j$. The gray region is $Q$. The left dotted polygonal chain is $\pi_k^l$. The red solid polygonal chain is in $B_1$ and the blue solid polygonal chain is in $B_2$.}
\label{fig:decomsubpocket}
\end{figure}

\begin{figure}[t]
\centering
\includegraphics[height=2.3in]{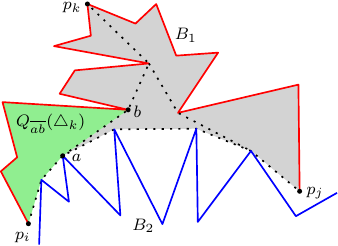}
\caption{Illustrating the case $a\neq p_j$. The green region is $Q_{\overline{ab}}(\triangle_k)$. The union of the green region and the gray region is $Q$. The red solid polygonal chain is in $B_1$ and the blue solid polygonal chain is in $B_2$.}
\label{fig:decomsubpocket20}
\end{figure}

Consider the left side of $\triangle_k$, denoted by $\pi_k^l$. Let $a$ be the leftmost vertex of the base side of $\triangle_k$. 
Depending on whether $a=p_i$, there are two cases. 

\begin{itemize}
    \item If $a=p_i$, then no edge of $\pi_k^l$ is a bridge and every edge of $\pi_k^l$ connects two vertices of $B_1$; see Figure~\ref{fig:decomsubpocket}. In this case, for each diagonal $d$ of $\pi_k^l$, all vertices of $Q_d(\triangle_k)$ are in $B_1(Q)$, where $Q_d(\triangle_k)$ is the subpolygon of $Q$ divided by $d$ that does not contain $\triangle_k$. Hence, $Q_d(\triangle_k)$ has no more than $n/2$ edges other than $d$. We create a child $v$ for $v^*$ in $\Psi$ with $P_v=Q_d(\triangle_k)$ and gate $d_v=d$. Since $P_v$ does not contain $d_{v^*}$, all three key properties hold for $v$. In this case, the processing of the left side $\pi_k^l$ is done. 
    
\item 
If $a\neq p_i$, then $a$ is in $B_2$ and the edge of $\pi_k^l$ incident to $a$ must be a bridge (see Figure~\ref{fig:decomsubpocket20}); let $b$ denote the other vertex of the edge. Note that none of the other edges of $\pi_k^l$ is a bridge. For each diagonal $d$ of $\pi_k^l$ with $d\neq \overline{ab}$, we create a child $v$ for $v^*$ in the same way as above. For $\overline{ab}$, the pocket $Q_{\overline{ab}}(\triangle_k)$ now also contains vertices of $B_2$ (see Figure~\ref{fig:decomsubpocket20}), and thus we cannot create a child for $v^*$ in the same way as above. Instead, we process $Q_{\overline{ab}}(\triangle_k)$ recursively. The reason we can do this is the following. By definition, the union of $\pi(p_i,a)$, which is the portion of $\pi(Q)$ between $p_i$ and $a$, and $\overline{ab}$ is a subpath of $\pi(p_i,p_k)$ and thus it is the geodesic path $\pi(p_i,b)$. In fact, $\overline{ab}$ is tangent to $\pi(Q)$ at $a$. Since $Q_{\overline{ab}}(\triangle_k)$ is a simple polygon with $\pi(p_i,b)$ on its boundary, $\pi(p_i,b)$ is also a concave chain and we can consider $\pi(p_i,b)$ as the base of $Q_{\overline{ab}}(\triangle_k)$ (note that $\pi(p_i,b)$ also contains exactly two bridges at the two ends). Hence, $Q_{\overline{ab}}(\triangle_k)$ has the same ``structure'' as $Q$ and thus we can apply the above processing recursively. Note that the number of vertices of $B_1$ in $Q_{\overline{ab}}(\triangle_k)$ is at most $|B_1(Q)|/2$. 
\end{itemize}

The right side of $\triangle_k$ can be processed symmetrically. 

The above process can be better represented by a tree $T_Q$. Each node $\nu$ of $T_Q$ corresponds to a subpolygon $Q(\nu)$ of $Q$ and also a non-empty geodesic triangle $\triangle_{\nu}$, which is added to the core $\core(v^*)$ of $v^*$ for $\Psi$. If $\nu$ is the root, then we have $Q(\nu)=Q$ and $\triangle_{\nu}=\triangle_k$. In the above processing of $\pi_k^l$, if $\pi_k^l$ has a bridge $\overline{ab}$, then the root of $T_Q$ has a left subtree by processing $Q_{\overline{ab}}(\triangle_k)$ recursively (if $\nu$ is the left child of the root, then $Q(\nu)=Q_{\overline{ab}}(\triangle_k)$). The right subtree is defined similarly. The tree height is $O(\log n)$, or more precisely $O(\log |B_1(Q)|)$ since every time we pick the middle vertex of the portion of $B_1$ in the current subpolygon $Q(\nu)$ to further decompose $Q(\nu)$. 
Note that as the base case, if $B_1(Q)$ contains a single vertex $p_k$ other than $p_i$ and $p_j$, then the two non-base sides of $\triangle_k$ are line segments $\overline{p_kp_i}$ and $\overline{p_kp_j}$, which are edges of $P$, and therefore the recursive step stops at $Q$. 

\subparagraph{Summary.}
The above defines the children $v$ of $v^*$ in the decomposition tree $\Psi$. As mentioned before, the entire tree $\Psi$ can be obtained by applying the same decomposition recursively on $P_v$ for every child $v$ of the root $v^*$. Since our decomposition does not introduce any new vertices and no two edges of the geodesic paths in the decomposition cross each other, the total size of the decomposition is $O(n)$. 

The above discusses the case where the geodesic triangle $\triangle^*$ is not empty. If it is empty, then it becomes a special case and can be handled similarly. Indeed, in this case, either $p_n$ is in $\pi(p^*,p_1)$ or $p_1$ is in $\pi(p^*,p_n)$. 
We let $\pi^*$ be $\pi(p^*,p_1)$ in the first case and $\pi(p^*,p_n)$ in the second case. Hence, $d_{v^*}=\overline{p_1p_n}\subseteq \pi^*$ holds in either case. We can simply treat the entire $P$ as $\calQ$, and then apply the above decomposition. Since $d_{v^*}$ is in $\pi^*$, according to our decomposition, it must lie on the base side of a geodesic triangle that is added to the core $\core(v^*)$. Hence, $d_{v^*}$ can never be in $P_v$ for any child $v$ of $v^*$. This ensures the second key property. The other two key properties follow from the same analysis. 


As computing a geodesic path in $P$ can be done in linear time~\cite{ref:GuibasLi87}, a straightforward algorithm can decompose each open subpocket $Q$ in $O(|Q|\log n)$ time since the height of $T_Q$ is $O(\log n)$. As such, decomposing all subpockets of $\calQ$ takes $O(|\calQ|\log n)$ time, and thus computing all the children of the root $v^*$ of the decomposition tree $\Psi$ can be done in $O(n\log n)$ time. As the height of $\Psi$ is $O(\log n)$ and the subpolygons $P_v$ of all nodes $v$ in the same level of $\Psi$ are interior disjoint, the total time for computing $\Psi$ can be bounded by $O(n\log^2 n)$. 

\subsection{Proving Theorem~\ref{theo:treequeries}: Constructing the visibility query data structure for $\boldsymbol{\Psi}$}
\label{sec:theo:treequeries}
In this section, using the decomposition tree $\Psi$, we construct a data structure for Theorem~\ref{theo:treequeries}. 

We start with constructing a point location data structure on the cores of all nodes of $\Psi$, which together form a decomposition of $P$. This takes $O(n)$ additional time~\cite{ref:KirkpatrickOp83,ref:EdelsbrunnerOp86}. Given a query point $q$, we can  use a point location query to find the node $v$ whose core $\core(v)$ contains $q$ in $O(\log n)$ time. 

In what follows, we explain the data structure $\calD_{v^*}$ for $P_{v^*}$ of the root node $v^*$ of $\Psi$, in which case $P_{v^*}=P$. The data structures for $P_v$ of other nodes $v$ can be built analogously. The data structure $\calD_{v^*}$ is used for answering the following queries: given a query point $q$, if $q\in \core(v^*)$, then $\vis(P_{v^*},q)$ can be computed in $O(\log n)$ time, and if $q$ is in $P_v$ for a child $v$ of $v^*$, then $\vis(P_{v^*}\setminus P_v,q)$ can be computed in $O(\log n)$ time.

\subsubsection{Overview}
Recall that the core $\core(v^*)$ is a set of geodesic triangles including $\triangle_{v^*}$. For the case $q\in \triangle_{v^*}$, we can simply build the data structure of Lemma~\ref{lem:geotriangle} with respect to $\triangle_{v^*}$, after which $\vis(P_{v^*},q)$ can be computed in $O(\log n)$ time for any $q\in \triangle_{v^*}$. Also, if $q\in P_v$ for a child $v$ of $v^*$ whose gate is a  diagonal on the boundary of $\triangle_{v^*}$ (we say that $v$ is {\em associated with} $\triangle_{v^*}$), then $\vis(P_{v^*}\setminus P_v,q)$ can be computed in $O(\log n)$ time by Lemma~\ref{lem:geotriangle}. 
The preprocessing of the above for $\triangle_{v^*}$ takes $O(n^2)$ space and $O(n^2\log n)$ time in total. 

It remains to handle the case when the query point $q$ is in $\calQ$. Suppose $q$ is in a subpocket $Q$ of $\calQ$.
By Observation~\ref{obser:vissubpocket}, to determine $\vis(P_{v^*},q)$, it suffices to consider $Q$ and $Q$'s neighboring subpockets. For each open subpocket $Q$ of $\calQ$, we preprocess $Q$ and its neighboring subpockets in the following. We assume that $Q$ is $B_1$-bounded and follow our previous notation.

Recall that we have a tree $T_Q$ for $Q$ and each node $\nu$ of $T_Q$ corresponds to a subpolygon $Q(\nu)$ and a geodesic triangle $\triangle_{\nu}$, which belongs to the core $\core(v^*)$. Furthermore, for each diagonal $d$ on a non-base side of $\triangle_{\nu}$ such that $d$ is not a bridge, $v^*$ has a child $v$ with $P_v=Q_d(\triangle_{\nu})$ and $d_v=d$. We say that $v$ is {\em associated with} $\triangle_{\nu}$. In addition, for each diagonal $d$ on the base of $Q(\nu)$ such that $d$ is on $\pi^*$ and $d$ is not a bridge, $v^*$ also has a child $v$ with $P_v$ being the subpolygon of $\calQ$ divided by $d$ that does not contain $Q$ and $d_v=d$. We also say that $v$ is {\em associated with} $\triangle_{\nu}$. 
Our goal for the preprocessing of $Q$ and its neighboring subpockets is to compute $\vis(P_{v^*},q)$ in $O(\log n)$ time if $q$ is in $\triangle_{\nu}$ for some node $\nu$ of $T_Q$, and compute $\vis(P_{v^*}\setminus P_v, q)$ in $O(\log n)$ time if $q$ is in $P_v$ for a child $v$ of $v^*$ associated with $\triangle_{\nu}$ for some node $\nu$ of $T_Q$. 

Define $\overline{Q}$ to be the union of $Q$ and all its neighboring subpockets of $P$ (including the special subpocket $Q^*=P_{v^*}\setminus \calQ$). Let $m=|B_1(Q)|$ and $M=|\overline{Q}|$. 

For each node $\nu$ of $T_Q$, we construct the data structure of Lemma~\ref{lem:geotriangle} with respect to $\triangle_{\nu}$. 
For each diagonal $d$ of $\partial \triangle_{\nu}$, recall that the notation $P_d(\triangle_{\nu})$ refers to the subpolygon of $P$ divided by $d$ that does not contain $\triangle_{\nu}$. 
Let $n_d$ denote the number of the elements (i.e., vertices or edges) of $P_d(\triangle_{\nu})$ that are visible to $d$, and let 
$\calC_d$ denote the set of visibility cones defined by these elements through $d$ (the definition of the cones follows that in Section~\ref{sec:smallspace}). According to the preprocessing algorithm of Lemma~\ref{lem:geotriangle}, the preprocessing takes $O(\sum_{d\in \partial \triangle_{\nu}}n_d^2)$ space and $O(\sum_{d\in \partial \triangle_{\nu}}n_d^2\cdot \log n)$ time in total. 

Therefore, to bound the total preprocessing time and space for the data structures of Lemma~\ref{lem:geotriangle} for all nodes $\nu\in T_Q$, it suffices to bound $\sum_{\nu\in T_Q}\sum_{d\in \partial \triangle_{\nu}}n_d^2$. To this end, we will prove $\sum_{v\in T_Q}\sum_{d\in \partial \triangle_{\nu}}n_d=O(M\log^2 n)$, which leads to $\sum_{v\in T_Q}\sum_{d\in \partial \triangle_{\nu}}n_d^2=O(M^2\log^4 n)$. In addition, we will also show that the corresponding sets $\calC_d$ of visibility cones can be computed in $O(M^2\log^4 n)$ time. As such, the preprocessing for $\overline{Q}$ takes $O(M^2\log^4n)$ space and $O(M^2\log^5 n)$ time in total. 

\subparagraph{Outline.}
For each node $\nu$ of $T_Q$, define $\Pi_{\nu}$ to be the union of the non-base sides of $\triangle_{\nu}$, and let $\pi_{\nu}$ be the base side of $\triangle_{\nu}$. 
In Section~\ref{sec:nonbase}, we first prove $\sum_{\nu\in T_Q}\sum_{d\in \Pi_{\nu}}n_d=O(m\log n)$ and also show that the corresponding visibility cones can be computed in $O(m\log n)$ time. 
Note that $n_d=O(|P_d(\triangle_{\nu})|)$ by definition. 
In Section~\ref{sec:baseside}, we prove the bound $\sum_{\nu\in T_Q}\sum_{d\in \pi_{\nu}}n_d=O(M\log^2 n)$ for the base sides $\pi_{\nu}$ and show that the corresponding visibility cones can be computed in $O(M\log^2 n)$ time.
Section~\ref{sec:summary} finally summarizes everything.

\subsubsection{The non-base sides}
\label{sec:nonbase}
Consider the root $\nu$ of $T_Q$. Following our previous notation, $\triangle_{\nu}$ is $\triangle_k$, formed by $\pi(p_i,p_k)$, $\pi(p_k,p_j)$, and $\pi(p_i,p_j)$. Recall that $\pi(p_i,p_j)$ is the base of $Q$. 
We argue that $\sum_{d\in \Pi_{\nu}}n_d=O(m)$. Note that the subpolygons $P_d(\triangle_{\nu})$ for all diagonals $d\in \Pi_{\nu}$ are interior disjoint. 
If $d$ is not a bridge, then all vertices of $P_d(\triangle_{\nu})$ are in $B_1(Q)$. Hence, the sum of $n_d$ for all such non-bridge diagonals of $\Pi_{\nu}$ is $O(m)$. Since each such $P_d(\triangle_{\nu})$ is a simple polygon, as discussed in Section~\ref{sec:smallspace}, the visibility cones for these diagonals can also be computed in $O(m)$ time~\cite{ref:GuibasLi87}. 

Next, consider a bridge $d\in \Pi_{\nu}$. Without loss of generality, we assume that $d$ is in the left side of $\triangle_{\nu}$. Following our previous notation, $d$ is $\overline{ab}$. 
In this case, $P_d(\triangle_{\nu})$ contains vertices of both $B_1$ and $B_2$, and $P_d(\triangle_{\nu})$ is not $Q_d(\triangle_{\nu})$ (because other than $\overline{ab}$, $Q_d(\triangle_{\nu})$ has another bridge); see Figure~\ref{fig:baseside10}. 

\begin{figure}[t]
\centering
\includegraphics[height=2.5in]{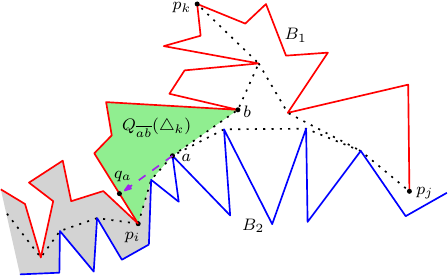}
\caption{The green region is $Q_{\overline{ab}}(\triangle_{\nu})$. The union of the green region and the gray region is $P_d(\triangle_{\nu})$. The red solid polygonal chain is in $B_1$ and the blue polygonal chain is in $B_2$. The portion of $B_1$ from $p_i$ to $b$ is $B_1(Q_{\overline{ab}}(\triangle_{\nu}))$.}
\label{fig:baseside10}
\end{figure}

We claim that all elements of $P_d(\triangle_{\nu})$ visible to any point of $\triangle_{\nu}$ through $\overline{ab}$ are in $B_1(Q_d(\triangle_{\nu}))$, i.e., the portion of $B_1$ in $Q_d(\triangle_{\nu})$. 
Indeed, without loss of generality, we assume that $P_d(\triangle_{\nu})$ is locally above $\overline{ab}$. Then, if we shoot a ray from $a$ along the direction opposite from $b$ until the first point $q_a$ on $\partial P_d(\triangle_{\nu})$ hit by the ray (see Figure~\ref{fig:baseside10}), then all elements of $P_d(\triangle_{\nu})$ visible to any point of $\triangle_{\nu}$ through $d$ must be in $\tilde{P_d}(\triangle_{\nu})$, where $\tilde{P_d}(\triangle_{\nu})$ is the subpolygon of $P_d(\triangle_{\nu})$ divided by $\overline{aq_a}$ and locally above $\overline{aq_a}$. By definition, $\pi(p_i,a)\cup \overline{ab}$ is a concave chain and is the base of $Q_d(\triangle_{\nu})$, and thus elements of $\tilde{P_d}(\triangle_{\nu})$ are all in $B_1(Q_d(\triangle_{\nu}))$. The claim thus follows. 

The claim implies that $n_d$ is bounded by the number of elements of $B_1$ in $Q_d(\triangle_{\nu})$ that are visible to $\overline{ab}$. Since $|B_1(Q_d(\triangle_{\nu}))|=O(m)$, we obtain $n_d=O(m)$. To compute the elements of $B_1(Q_d(\triangle_{\nu}))$ visible to $\overline{ab}$ and their corresponding visibility cones, we can do the following. First, we compute the point $q_a$ by a ray-shooting query in $O(\log n)$ time~\cite{ref:ChazelleRa94,ref:HershbergerA95} (after $O(n)$-time preprocessing for $P$). Then, we can obtain $\tilde{P_d}(\triangle_{\mu})$ in $O(m)$ time. Next, we compute all visibility cones using the simple polygon $\tilde{P_d}(\triangle_{\mu})$ in $O(m)$ time~\cite{ref:GuibasLi87}. In this way, in $O(m+\log n)$ time we can compute all visibility cones for $d=\overline{ab}$. 

The above shows that $\sum_{d\in \Pi_{\nu}}n_d=O(m)$ and the corresponding visibility cones can be computed in $O(m+\log n)$ time. 

Now consider the left child $\nu'$ of $\nu$ in $T_Q$, which is defined by the bridge $\overline{ab}$ of the left children side of $\triangle_{\nu}$, i.e., $Q(\nu')=Q_{\overline{ab}}(\triangle_{\nu})$; see Figure~\ref{fig:baseside10}. Following the same analysis as above but with respect to $Q_{\overline{ab}}(\triangle_{\nu})$, we can show that
$\sum_{d\in \Pi_{\nu'}}n_d=O(m')$ and the corresponding visibility cones can be computed in $O(m'+\log n)$ time, where $m'$ is the number of vertices of $B_1$ in $Q_{\overline{ab}}(\triangle_{\nu})$ and $m'\leq m/2$ by definition. Indeed, we are able to follow the same analysis as above for $\nu$ because as discussed before $Q_{\overline{ab}}(\triangle_{\nu})$ has the same structure as $Q$, i.e., the base of $Q_{\overline{ab}}(\triangle_{\nu})$ is a concave chain with two bridges at the two ends. Following the same analysis recursively, since the height of $T_Q$ is $O(\log m)$ (and the sum of $m'$ for all nodes in the same level is $O(m)$), we can obtain that $\sum_{\nu\in T_Q}\sum_{d\in \Pi_{\nu}}n_d=O(m\log m)=O(m\log n)$ and the corresponding visibility cones can be computed in $O(m\log n)$ time.

\subsubsection{The base sides}
\label{sec:baseside}
We now prove the bound $\sum_{\nu\in T_Q}\sum_{d\in \pi_{\nu}}n_d=O(M\log^2 n)$ for the base sides $\pi_{\nu}$ and show that the corresponding visibility cones can be computed in $O(M\log^2 n)$ time.

For each node $\nu\in T_Q$, recall that $\pi_{\nu}$ has at most two bridges, which are at the two ends of $\pi_{v}$, and all other edges are on $\pi(Q)$, the base of $Q$. The geodesic triangles $\triangle_{\nu}$ of all $\nu\in T_Q$ are interior disjoint, and thus each edge of $\pi(Q)$ appears in $\pi_{\nu}$ for exactly one node $\nu\in T_Q$. For each diagonal $d$ on $\pi(Q)$, suppose that $d$ is in $\pi_{\nu}$ for some node $\nu$ of $T_Q$. Let $\overline{Q}_d(\triangle_{\nu})$ be the subpolygon of $\overline{Q}$ divided by $d$ that does not contain $\triangle_{\nu}$. By definition, $\overline{Q}_d(\triangle_{\nu})$ is a $B_2$-bounded subpocket of $\calQ$ neighboring to $Q$ through $d$ and $P_d(\triangle_{\nu})=\overline{Q}_d(\triangle_{\nu})$. Hence, $n_d=O(|\overline{Q}_d(\triangle_{\nu})|)$ and the corresponding visibility cones can be computed in $O(|\overline{Q}_d(\triangle_{\nu})|)$ time. As such, the sum of $n_d$ for all non-bridge diagonals $d$ in $\pi_{\nu}$ for all nodes $\nu\in T_Q$ is bounded by the total size of all neighboring subpockets of $Q$, which is $O(M)$. Also, the corresponding visibility cones can be computed in $O(M)$ time. 

It remains to bound the sum of $n_d$ for bridges $d\in \pi_{\nu}$ for all nodes $\nu\in T_Q$. 

For each node $\nu\in T_Q$, let $\overline{Q}(\nu)$ denote the union of $Q(\nu)$ and the neighboring subpockets of $Q(\nu)$. Again, the base side $\pi_{\nu}$ of $\triangle_{\nu}$ has at most two bridges, which are the leftmost and rightmost edges of $\pi_{\nu}$, respectively. If the leftmost edge of $\pi_{\nu}$ is a bridge, we call it the {\em left bridge} and denote it by $l_{\nu}$. If the rightmost edge of $\pi_{\nu}$ is a bridge, we call it the {\em right bridge} and denote it by $r_{\nu}$. 
In that follows, we prove $\sum_{\nu\in T_Q}n_{l_{\nu}}=O(M\log^2 n)$, and we can prove the same bound for the right bridges $r_{\nu}$. We will also show that the corresonding visibility cones can be computed in $O(M\log^2 n)$ time. 

\subparagraph{Proving the bound $\boldsymbol{\sum_{\nu\in T_Q}n_{l_{\nu}}=O(M\log^2 n)}$.}
For each node $\nu\in T_Q$, define $Z(\nu)$ as the set of elements of $\overline{Q}$ visible to any point of $\triangle_{\nu}$ through the left bridge $l_{\nu}$ of $\pi_{\nu}$. In other words, $Z(\nu)$ is the set of elements of $\overline{Q}_{l_{\nu}}(\triangle_{\nu})$ visible to at least one point of $l_{\nu}$, where $\overline{Q}_{l_{\nu}}(\triangle_{\nu})$ is the subpolygon of $\overline{Q}$ divided by $l_{\nu}$ that does not contain $\triangle_{\nu}$. By definition, we have $Z(\nu)\subseteq \overline{Q}_{l_{\nu}}(\triangle_{\nu})$. 
Note that if $\triangle_{\nu}$ does not have a left bridge $l_{\nu}$, then $Z(\nu)=\emptyset$. One reason we introduce $Z(\nu)$ is that $n_{l_{\nu}}=O(|Z(\nu)|)$ holds. Hence, to bound $n_{l_{\nu}}$, we can bound $O(|Z(\nu)|)$ instead. 

In addition, we will define a subpolygon $A(\nu)\subseteq \overline{Q}$ with the following invariant (called {\em $A$'s invariant}): $Z(\mu)\subseteq A(\nu)$ holds for $\mu=\nu$ and any node $\mu$ in the right subtree of $\nu$. 
Hence, to bound $|Z(\mu)|$, we can bound $A(\nu)$. To prove $\sum_{\nu\in T_Q}n_{l_{\nu}}=O(M\log^2 n)$, we will prove $\sum_{\nu\in T_Q}|A(\nu)|=O(M\log^2 n)$. 
As will be seen, $A(\nu)$ is either $\overline{Q}(\nu)$ or $A(\nu')$ for some ancestor $\nu'$ of $\nu$ such that $\nu$ is in the right subtree of $\nu'$.

For the root $\hat\nu$ of $T_Q$, we let $A(\hat\nu)=\overline{Q}$, which is also $\overline{Q}(\hat \nu)$, and thus $A$'s invariant obviously holds for $\hat\nu$. We next consider a child $\nu$ of $\hat\nu$. There are two cases depending on whether $\nu$ is the left or right child of $\hat\nu$. 

\begin{description}
\item[Right child case.] 
We first assume $\nu$ is the right child of $\hat\nu$, i.e., the vertices of $B_1$ in $Q(\nu)$ are between $p_k$ and $p_j$ (see Figure~\ref{fig:baseside20}). 
Depending on whether $\pi_{\nu}$ has a left bridge, there are two subcases. 

\begin{figure}[t]
\centering
\includegraphics[height=2.5in]{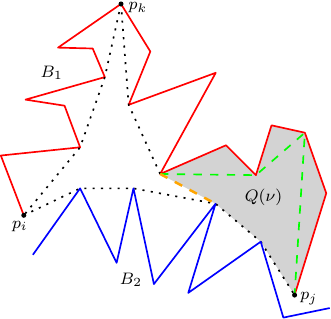}
\caption{The gray region is $Q(\nu)$ for the right child $\nu$ of $\hat\nu$. The two non-base sides of $\triangle_{\nu}$ are marked with green dashed segments. The orange dashed segment on the base side $\pi_{\nu}$ of $\triangle_{\nu}$ is the left bridge $l_{\nu}$ of $\pi_{\nu}$. In this case, $\nu$ cannot have a left subtree in $T_Q$.}
\label{fig:baseside20}
\end{figure}

\begin{itemize}
    \item 
If $\pi_{\nu}$ has a left bridge (see Figure~\ref{fig:baseside20}), then a crucial observation is that $\nu$ does not have a left subtree because the left non-base side of $\triangle_{\nu}$ cannot have a bridge according to our decomposition of $Q$. In this case, we define $A(\nu)=A(\hat\nu)$, which is $\overline{Q}$. Hence, $A$'s invariant clearly holds for $\nu$.

\item 
If $\pi_{\nu}$ does not have a left bridge (see Figure~\ref{fig:baseside30}), then $Z(\nu)=\emptyset$ by definition. In this case, the key observation is that for any node $\mu$ in the right subtree of $\nu$, $Z(\mu)\subseteq \overline{Q}(\nu)$ must hold. To see this, we give a brief explanation here and a more detailed argument is given in the proof of Lemma~\ref{lem:rightsubtree}. 
Indeed, since $\pi_{\nu}$ does not have a left bridge, the left side of $\triangle_{\nu}$ must contain a bridge $d$ that is tangent to the base of $Q$ (see Figure~\ref{fig:baseside30}). Also, the right side of $\triangle_{\hat\nu}$ has a bridge $d'$ (in fact, $Q(\nu)$ is $Q_{d'}(\triangle_{\hat\nu})$). The two bridges $d$ and $d'$ along with the portion of the base of $Q$ between them is a geodesic path (specifically, it is a subpath of the geodesic path containing the left side of $\triangle_{\nu}$; see Figure~\ref{fig:baseside30}). For any point $q$ in $\triangle_{\mu}$ for any node $\mu$ in the right subtree of $\nu$, suppose that $q$ is visible to a point $p\in \overline{Q}\setminus \overline{Q}(\nu)$ through the left bridge $l_{\mu}$ of $\triangle_{\mu}$ (i.e., $\overline{pq}$ intersects $l_{\mu}$). Then, $\overline{pq}$ must cross both $d$ and $d'$, which is not possible since a line segment cannot cross a geodesic path twice. Hence, we have $Z(\mu)\subseteq \overline{Q}(\nu)$. 

\begin{figure}[t]
\centering
\includegraphics[height=2.2in]{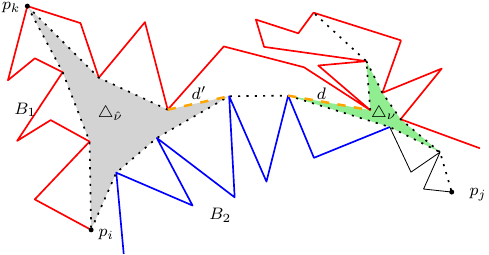}
\caption{Illustrating the geodesic triangles $\triangle_{\hat\nu}$ (the gray region) and $\triangle_{\nu}$ (the green region). For any $\mu$ in the right subtree of $\nu$, $\triangle_{\mu}$ must be ``right'' of $\triangle_{\nu}$.}
\label{fig:baseside30}
\end{figure}

In light of the observation, we set $A(\nu)=\overline{Q}(\nu)$. This means that the right subtree of $\nu$ can be treated independently with respect to $\overline{Q}(\nu)$. 
\end{itemize}

\item[Left child case.]     
If $\nu$ is the left child of $\hat\nu$, then for any node $\mu$ in the right subtree of $\nu$, 
according to our decomposition, $\overline{Q}_{l_{\mu}}(\triangle_{\mu})\subseteq \overline{Q}(\nu)$. Since $Z(\mu)\subseteq \overline{Q}_{l_{\mu}}(\triangle_{\mu})$ by definition, we obtain $Z(\mu)\subseteq \overline{Q}(\nu)$. Therefore, we set $A(\nu)=\overline{Q}(\nu)$, meaning that the right subtree of $\nu$ can be treated independently with respect to $\overline{Q}(\nu)$. 
\end{description}

In general, consider any other node $\nu$ in the tree $T_Q$. Assume that $A(\nu')$ for all ancestors $\nu'$ of $\nu$ have been set and $A$'s invariant holds for $\nu'$. 
We define $A(\nu)$ as follows. 

If $\pi_{\nu}$ has a left bridge $l_{\nu}$, then $l_{\nu}$ must be on the right side of $\triangle_{\nu'}$ of the lowest ancestor $\nu'$ of $\nu$ with $\nu$ in the right subtree of $\nu'$. In this case, we set $A(\nu)=A(\nu')$. To see that $A$'s invariant holds for $\nu$, consider a node $\mu$ that is either $\nu$ or in the right subtree of $\nu$. Since $\nu$ is in the right subtree of $\nu'$, $\mu$ is also in the right subtree of $\nu'$. Since $A$'s invariant holds for $\nu'$, we have $Z(\mu)\subseteq A(\nu')$. As $A(\nu)=A(\nu')$, we obtain $Z(\mu)\subseteq A(\nu)$. Hence, $A$'s invariant holds for $\nu$. 

If $\pi_{\nu}$ does not have a left bridge, then $Z(\nu)=\emptyset$. 
In this case, we set $A(\nu)=\overline{Q}(\nu)$. The following lemma guarantees that $A$'s invariant holds for $\nu$. 
\begin{lemma}\label{lem:rightsubtree}
If $\pi_{\nu}$ does not have a left bridge, then for any node $\mu$ in the right subtree of $\nu$, $Z(\mu)\subseteq \overline{Q}(\nu)$. 
\end{lemma}
\begin{proof}
Since $\pi_{\nu}$ does not have a left bridge, according to our decomposition, the left side of $\triangle_{\nu}$ must contain a bridge $d$, which is tangent to the base of $Q$. 

Let $\nu'$ be the lowest ancestor of $\nu$ such that $\nu$ is in the right subtree of $\nu'$. According to our decomposition, the right side of $\triangle_{\nu'}$ has a bridge $d'$ (e.g., in Figure~\ref{fig:baseside30}, one can replace $\hat\nu$ by $\nu'$). 
The two bridges $d$ and $d'$ along with the portion of the base of $Q$ between them is a geodesic path (specifically, it is a subpath of the geodesic path that contains the left side of $\triangle_{\nu}$). 

Consider a node $\mu$ in the right subtree of $\nu$. By definition, $Z(\mu)\subseteq \overline{Q}_{l_{\mu}}(\triangle_{\mu})$. According to our decomposition, since $\mu$ is in the right substree of $\nu$, we have $\overline{Q}_{l_{\mu}}(\triangle_{\mu})\subseteq \overline{Q}(\nu)\cup \overline{Q}'_{d'}(\triangle_{\nu'})$, where $\overline{Q}'_{d'}(\triangle_{\nu'})$ is the subpolygon of $\overline{Q}$ divided by $d'$ containing $\triangle_{\nu'}$. In addition, since $d$ is a bridge in the left  side of $\triangle_{\nu}$ and $\mu$ is in the right subtree of $\nu$, by the definition of $v'$, $d$ separates $\overline{Q}(\mu)$ from $\overline{Q}'_{d'}(\triangle_{\nu'})$ and thus also separates $\triangle_{\mu}$ from $\overline{Q}'_{d'}(\triangle_{\nu'})$ since $\triangle_{\mu}\subseteq \overline{Q}(\mu)$. According to our decomposition, $\triangle_{\mu}\subseteq \overline{Q}(\mu)\subseteq\overline{Q}(\nu)\subseteq \overline{Q}_{d'}(\triangle_{\nu'})$ and thus $d'$ also separates $\triangle_{\mu}$ from $\overline{Q}'_{d'}(\triangle_{\nu'})$. Hence, both $d$ and $d'$ separate $\triangle_{\mu}$ from $\overline{Q}'_{d'}(\triangle_{\nu'})$. This means that if there is a point $p\in \overline{Q}'_{d'}(\triangle_{\nu'})$ visible to a point $q\in \triangle_{\mu}$, then $\overline{pq}$ must cross both $d$ and $d'$. But this is impossible since $d$ and $d'$ are in a geodesic path.
Therefore, no point of $\overline{Q}'_{d'}(\triangle_{\nu'})$ can be visible to any point of $\triangle_{\mu}$, implying that $\overline{Q}'_{d'}(\triangle_{\nu'})\cap Z(\mu)=\emptyset$. Since  $Z(\mu)\subseteq \overline{Q}_{l_{\mu}}(\triangle_{\mu})\subseteq \overline{Q}(\nu)\cup \overline{Q}'_{d'}(\triangle_{\nu'})$, we obtain  $Z(\mu)\subseteq \overline{Q}(\nu)$. 

The above discusses the case where $\nu'$ exists. If not, then we simply have $\overline{Q}_{l_{\mu}}(\triangle_{\mu})\subseteq \overline{Q}(\nu)$. Therefore, we also obtain $Z(\mu)\subseteq\overline{Q}_{l_{\mu}}(\triangle_{\mu})\subseteq \overline{Q}(\nu)$. 
\end{proof}

Since $Z(\nu)\subseteq A(\nu)$ for each node $\nu$ of $T_Q$, we have $\sum_{\nu\in T_Q}n_{l_{\nu}}=O(\sum_{\nu\in T_Q}|A(\nu)|)$. Also, for each $\nu$, the corresponding visibility cones for $l_{\nu}$ can be found in $A(\nu)$, which takes $O(|A(\nu)|)$ time~\cite{ref:GuibasLi87} since $A(\nu)$ is a simple polygon. Hence, computing all visibility cones takes $O(\sum_{\nu\in T_Q}|A(\nu)|)$ time. 
The following lemma gives a bound for $\sum_{\nu\in T_Q}|A(\nu)|$. 

\begin{lemma}\label{lem:Abound}
    $\sum_{\nu\in T_Q}|A(\nu)|=O(M\log^2 n)$.
\end{lemma}
\begin{proof}
We partition all nodes $\nu\in T_Q$ into two categories: (1) $A(\nu)=\overline{Q}(\nu)$; (2) $A(\nu)\neq \overline{Q}(\nu)$. For the first category, it suffices to bound $\sum_{\nu}|\overline{Q}(\nu)|$. Observe that $\overline{Q}(\nu)$ for all nodes $\nu$ at the same level of $T_Q$ are interior disjoint, and therefore, the sum of their sizes is $O(M)$. As $T_Q$ has $O(\log m)$ levels, we obtain $\sum_{\nu}|\overline{Q}(\nu)|=O(M\log m)$. We next analyze the second category. 

If $A(\nu)\neq \overline{Q}(\nu)$, then by our way of defining $A(\cdot)$ for the nodes of $T_Q$, $A(\nu)$ must be equal to $A(\nu')=\overline{Q}(\nu')$ for some ancestor $\nu'$ such that $\nu$ is in the right subtree of $\nu'$. Furthermore, $\nu$ cannot have a left subtree. To see this, since $A(\nu)\neq \overline{Q}(\nu)$, $\pi_{\nu}$ has a left bridge that is also on the right  side of $\triangle_{\nu''}$ for an ancestor $\nu''$ of $\nu$ such that $\nu$ is in the right subtree of $\nu''$. This implies that the left side of $\triangle_{\nu}$ cannot have a bridge, and thus $\nu$ cannot have a left child.   

Using the above property, we claim that in the right subtree of $\nu'$ there is at most one node $\mu$ in each level such that $A(\mu)=A(\nu')$. Indeed, suppose to the contrary that there are two nodes $\mu_1$ and $\mu_2$ in the same level with $A(\mu_1)=A(\mu_2)=A(\nu')$. Let $\mu_0$ be the lowest common ancestor of $\mu_1$ and $\mu_2$. Without loss of generality, we assume that $\mu_1$ is in the left subtree of $\mu_0$ and $\mu_2$ is in the right subtree of $\mu_0$. Since $A(\mu_1)=A(\mu_2)=A(\nu')$, according to our way of defining $A(\cdot)$, both $\mu_1$ and $\mu_2$ must be in the right subtree of $\nu'$. Therefore, $\mu_0$ is also in the right subtree of $\nu'$. Since $A(\mu_2)=A(\nu')$ and $\mu_2$ is in the right subtree of $\mu_0$, according to our way of defining $A(\cdot)$, $A(\mu_0)$ must also be $A(\nu')$ (this is because for any node $\mu\in T_Q$, unless $A(\mu)=\overline{Q}(\mu)$, $A(\mu)$ is always equal to $A(\mu')$ of the lowest ancestor $\mu'$ of $\mu$ with $\mu$ in the right subtree of $\mu'$). By the property discussed above, $\mu_0$ cannot have a left subtree. But this contradicts the fact that $\mu_1$ is in the left subtree of $\mu_0$. 
The claim thus follows. 

Again by our way of defining $A(\cdot)$, no node $\mu$ in the left subtree of $\nu'$ can have $A(\mu)=A(\nu')$, which is $\overline{Q}(\nu')$.

The above implies that $\overline{Q}(\nu')$ will be included in $\sum_{\nu} |A(\nu)|$ at most $O(\log m)$ times since the height of $T_Q$ is $O(\log m)$. Hence, $\sum_{\nu} |A_{\nu}|$ is no more than $O(\log n)$ times $\sum_{\nu} |\overline{Q}(\nu)|$. As $\sum_{\nu} |\overline{Q}(\nu)|=O(M\log m)$, we obtain $\sum_{\nu} |A_{\nu}|=O(M\log^2 m)$, which is $O(M\log^2 n)$ since $m\leq n$. The lemma thus follows. 
\end{proof}

\subsubsection{Summary}
\label{sec:summary}
In summary, we can construct a data structure for an open subpocket $Q$ of $\calQ$ in $O(M^2\log^4 n)$ space and $O(M^2\log^5 n)$ time such that for any query point $q$, if $q$ is in $\triangle_{\nu}$ for some node $\nu$ of $T_Q$, then $\vis(P_{v^*},q)$ can be computed in $O(\log n)$ time, and if $q$ is in $P_v$ for a child $v$ of $v^*$ that is associated with $\triangle_{v^*}$ or $\triangle_{\nu}$ for a node $\nu\in T_Q$, then $\vis(P_{v^*}\setminus P_v,q)$ can be computed in $O(\log n)$ time. 

We construct the above data structure for all $B_1$-bounded open subpockets $Q$ of $\calQ$. Since all subpockets of $\calQ$ are interior disjoint and each $B_2$-bounded subpocket can be a neighboring subpocket of at most two $B_1$-bounded subpockets (and the special subpocket $Q^*$ can be a neighboring subpocket of at most one subpocket of $\calQ$), the sum of $M$ for all such $Q$'s is $O(n)$. Hence, the data structures for all $B_1$-bounded subpockets can be constructed in $O(n^2\log^4 n)$ space and $O(n^2\log^5 n)$ time. Similarly, we construct the same data structures for all $B_2$-bounded open subpockets, which also takes $O(n^2\log^4 n)$ space and $O(n^2\log^5 n)$ time. After that, for any query point $q$, if $q$ is in the core $\core(v^*)$, then $\vis(P_{v^*},q)$ can be computed in $O(\log n)$ time, and if $q$ is in $P_v$ for a child $v$ of $v^*$, then $v$ must be associated with $\triangle_{v^*}$ or $\triangle_{\nu}$ for a node $\nu\in T_Q$ for some open subpocket $Q$ of $\calQ$ and thus $\vis(P_{v^*}\setminus P_v,q)$ can be computed in $O(\log n)$ time. 

This finishes the preprocessing for the root $v^*$ of the decomposition tree $\Psi$. The children of the root can then processed recursively. In general, for each node $v\in \Psi$, for each geodesic triangle $\triangle$ of the core $\core(v)$, we build the data structure of Lemma~\ref{lem:geotriangle} for $\triangle$ in $P_{v}$. This takes $O(n_v^2\log^4 n_v)$ space and $O(n_v^2\log^5 n_v)$ preprocessing time, where $n_v=|P_v|$. As the sizes of the subpolygons $P_v$ decrease geometrically along any root-to-leaf path and the subpolygons $P_v$ for all nodes $v$ at the same level of $\Psi$ are interior disjoint, the total space of the data structures for the entire $\Psi$ is $O(n^2\log^ 4n)$ and the total preprocessing time is $O(n^2\log^5 n)$.


In summary, after $O(n^2\log^4 n)$ space and $O(n^2\log^5 n)$ time preprocessing, given a query point $q$, if $q$ is in the core $\core(v)$ for a node $v$ of $\Psi$, then $\vis(P_v,q)$ can be computed in $O(\log n)$ time, and if $q$ is in $P_v$ of a child $v$ of a node $u$, then $\vis(P_u\setminus P_v,q)$ can be computed in $O(\log n)$ time. This proves Theorem~\ref{theo:treequeries}.

\section{A data structure of optimal query time}
\label{sec:Optimal_query_time}

In this section, we present our visibility query data structure of $O(n^{2+\epsilon},n^{2+\epsilon}, \log n)$ complexity, based on the new decomposition introduced in Section~\ref{sec:decomp}.  

\subparagraph{Preprocessing.}
We first construct the decomposition tree $\Psi$ and build the data structure of Theorem~\ref{theo:treequeries}.  
This requires $O(n^2 \log^4 n)$ space and $O(n^2 \log^5 n)$ preprocessing time.  

For each node $v \in \Psi$, recall that all edges of the subpolygon $P_v$ are edges of $P$, except for the gate $d_v$.  
For simplicity, let $|P_v|$ denote the number of edges of $P$ that belong to $P_v$.  

Let $r$ be a parameter with $1 \le r \le n$.  
A node $v$ of $\Psi$ is called a {\em border node} if $|P_{v_p}| \ge n/r$ and $|P_v| < n/r$, where $v_p$ denotes the parent of $v$. By definition, each leaf-to-root path in $\Psi$ contains at most one border node.  
This further implies that the subpolygons $P_v$ corresponding to all border nodes $v$ are edge-disjoint.  

The border nodes, together with all their ancestors in $\Psi$, form a subtree $\Psi_b$ of $\Psi$.  
Note that $\Psi_b$ includes the root and has all border nodes as its leaves (hence, the border nodes indeed form the ``border'' of $\Psi_b$).  
We use $B(\Psi)$ to denote the set of all border nodes of $\Psi$.  
We further have the following lemma.

\begin{lemma}\label{lem:internalnumber}
The number of internal nodes of $\Psi_b$ is $O(r)$.  
\end{lemma}
\begin{proof}
Let $\Psi_b'$ denote the subtree of $\Psi_b$ obtained by removing all its leaves.  
Our goal is to show that $\Psi_b'$ contains $O(r)$ nodes.  

By definition, $|P_v| \ge n/r$ holds for every node $v$ of $\Psi_b'$.  
Since the edges of $P_v$ for all leaves $v \in \Psi_b'$ are disjoint, $\Psi_b'$ has $O(r)$ leaves.  
According to our decomposition, $|P_v| \le |P_{v_p}|/2$ holds for every non-root node $v \in \Psi$.  

For any value $m$, let $V_m$ denote the set of nodes $v \in \Psi_b'$ satisfying $m/2 < |P_v| \le m$.  
For any two nodes $u, v \in V_m$, it is impossible for one to be an ancestor of the other.  
Indeed, suppose to the contrary that $u$ is an ancestor of $v$.  
Then we would have $|P_u| \le |P_v|/2 \le m/2$, contradicting the condition $m/2 < |P_u|$.  
Hence, $P_u$ and $P_v$ are interior disjoint, which implies that $|V_m| \le 2n/m$.  

The total number of nodes in $\Psi_b'$ is therefore bounded by
\[
|V_n| + |V_{n/2}| + |V_{n/4}| + \cdots + |V_{n/r}| 
    \le 2 + 4 + 8 + \cdots + 2r = O(r).
\]
The lemma thus follows. 
\end{proof}

We further perform the following preprocessing.  
For each internal node $v \in \Psi_b$, we construct a data structure $\calD_v$ with respect to the gate $d_v$ such that, for any point $q \in P_v$, $\vis(P \setminus P_v, q)$ can be computed in $O(\log n)$ time.  
This corresponds to a diagonal-separated subproblem and requires $O(n^2)$ space and $O(n^2 \log n)$ preprocessing time~\cite{ref:AronovVi02}, as described in Section~\ref{sec:diasep}.  
Since $\Psi_b$ has $O(r)$ internal nodes by Lemma~\ref{lem:internalnumber}, constructing the data structures for all these nodes takes $O(n^2 r)$ space and $O(n^2 r\log n)$ preprocessing time in total.  

For each leaf $v \in \Psi_b$, we recursively preprocess $P_v$ and the subtree of $\Psi$ rooted at $v$.  
This completes the preprocessing phase.  

We now analyze the overall space.  
Let $S(n)$ denote the total space of the preprocessing (excluding the data structure of Theorem~\ref{theo:treequeries}).  
We have the following recurrence:
\[
S(n) = \sum_{v \in B(\Psi)} S(|P_v|) + O(n^2 r).
\]
Since $|P_v| < n/r$ for every node $v \in B(\Psi)$, and the subpolygons $P_v$ for all $v \in B(\Psi)$ are edge-disjoint (implying that $\sum_{v \in B(\Psi)} |P_v| \le n$),  
by setting $r = O(n^{\epsilon})$ for any constant $\epsilon > 0$, the recurrence solves to $S(n) = O(n^{2+\epsilon})$.  

Similarly, the preprocessing time is also $O(n^{2+\epsilon})$. Including the data structure of Theorem~\ref{theo:treequeries}, the total preprocessing requires $O(n^{2+\epsilon})$ space and time.

\subparagraph{Queries.}

Recall that the data structure of Theorem~\ref{theo:treequeries} includes a point location structure for the cores $\core(v)$ of the nodes $v \in \Psi$, which together form a decomposition of $P$.  

Given a query point $q$, we first perform a point location query to determine the node $v$ of $\Psi$ whose core $\core(v)$ contains $q$.  
Then, by Theorem~\ref{theo:treequeries}, we compute $\vis(P_v, q)$ in $O(\log n)$ time.  
Depending on whether $v$ is an internal node of $\Psi_b$, we distinguish two cases.  

\begin{itemize}
    \item 
If $v$ is an internal node of $\Psi_b$, then using the data structure $\calD_v$, we compute $\vis(P \setminus P_v, q)$ in $O(\log n)$ time.  
Finally, by Lemma~\ref{lem:merge}, we merge $\vis(P_v, q)$ and $\vis(P \setminus P_v, q)$ in $O(\log n)$ time to obtain $\vis(q)$.  
Thus, in this case, computing $\vis(q)$ takes $O(\log n)$ time in total.  

\item
If $v$ is not an internal node of $\Psi_b$, then by following the path from $v$ to the root of $\Psi$, we find a border node $v'$.  
Using the recursively constructed data structure for $P_{v'}$, we compute $\vis(P_{v'}, q)$ recursively.  
Let $u$ be the parent of $v'$.  
By definition, $u$ is an internal node of $\Psi_b$.  
By Theorem~\ref{theo:treequeries}, we compute $\vis(P_u \setminus P_{v'}, q)$ in $O(\log n)$ time, and by Lemma~\ref{lem:merge}, we merge $\vis(P_{v'}, q)$ and $\vis(P_u \setminus P_{v'}, q)$ in $O(\log n)$ time to obtain $\vis(P_u, q)$.  
Finally, using the data structure $\calD_u$, we compute $\vis(P \setminus P_u, q)$ in $O(\log n)$ time and merge it with $\vis(P_u, q)$, again in $O(\log n)$ time, to obtain $\vis(q)$.  

Let $Q(n)$ denote the query time.  
Since $|P_{v'}| < n/r$, the recursive computation of $\vis(P_{v'}, q)$ takes $Q(n/r)$ time.  
Hence, the recurrence relation is
\[
Q(n) = Q(n/r) + O(\log n),
\]
which solves to $Q(n) = O(\log n)$ for $r = n^{\epsilon}$, where $\epsilon$ is a constant.  
\end{itemize}

The following theorem summarizes our result.

\begin{theorem}\label{theo:optquery}
Given a simple polygon $P$ with $n$ vertices, a data structure of size $O(n^{2+\epsilon})$ can be constructed in $O(n^{2+\epsilon})$ time such that the visibility polygon $\vis(q)$ can be computed in $O(\log n)$ time for any query point $q \in P$.
\end{theorem}

\subparagraph{Remark.}
One may wonder whether the balanced decomposition $\calT(P)$ from Section~\ref{sec:bdscheme} could replace our decomposition $\Psi$ to achieve the same result.  
This, however, appears challenging because $\calT(P)$ does not satisfy the gate property.  
Indeed, each subpolygon in $\calT(P)$ may connect to the ``outside world'' through up to $O(\log n)$ diagonals.  
This limitation is precisely why we introduce a new decomposition for $P$.


\section{A quadratic-space data structure}
\label{sec:A_Quadratic_Space_Structure}

In this section, we present our data structure of complexity $O(n^2\log n,n^2,\log n\log\log n)$. 
We will first give a data structure of $O(n^2\log^5n,n^2\log^4 n,\log n\log\log n)$ using our decomposition $\Psi$. By integrating it with the framework in Section~\ref{sec:bdscheme}, we will reduce the space to $O(n^2)$ and reduce the preprocessing time to $O(n^2\log n)$, while keeping the same query time. 

In the preprocessing phrase, we compute the decomposition tree $\Psi$ and the data structure of Theorem~\ref{theo:treequeries}. This takes $O(n^2\log^4 n)$ space and $O(n^2\log^5n)$ preprocessing time.

\subparagraph{Super-levels.}
We partition $\Psi$ into $O(\log\log n)$ {\em super-levels}.  
Each super-level consists of a collection of subtrees of $\Psi$, and we refer to each such subtree as a {\em component subtree} of the super-level.  
The $i$-th super-level $\Psi_i$, with $1 \le i \le O(\log\log n)$, is defined as follows.  

We assume that the set $R_i$ of the roots of all component subtrees of $\Psi_i$ is already known and that the following {\em algorithm invariants} hold for $R_i$:  
(1) $|P_v| \le 2n / 2^{2^i}$ for every node $v \in R_i$;  
(2) the subpolygons $P_v$ for all $v \in R_i$ are edge-disjoint; and  
(3) for each node $v \in R_i$ with $i > 0$, its parent lies in the $(i-1)$-th super-level.  
Initially, when $i = 0$, the set $R_0$ consists of a single node that is the root of $\Psi$, and thus all invariants hold for $R_0$.  

We now define $\Psi_i$.  
Consider a node $v \in R_i$.  
Let $\Psi(v)$ denote the subtree of $\Psi$ rooted at $v$.  
We borrow several concepts from Section~\ref{sec:Optimal_query_time}.  
A non-root node $u$ of $\Psi(v)$ is called a {\em border node} of $\Psi(v)$ if $|P_{u_p}| \ge |P_v| / 2^{2^i}$ and $|P_u| < |P_v| / 2^{2^i}$, where $u_p$ is the parent of $u$.  
The border nodes, together with all their ancestors in $\Psi(v)$, form a subtree $\Psi_b(v)$ of $\Psi(v)$.  
Let $\Psi_i(v)$ denote the subtree of $\Psi_b(v)$ obtained by removing all its leaves, and let $B_v$ denote the set of border nodes of $\Psi(v)$.  
We then define
\[
\Psi_i = \bigcup_{v \in R_i} \Psi_i(v)
\quad \text{and} \quad
R_{i+1} = \bigcup_{v \in R_i} B_v.
\]

We have the following observation.  
\begin{observation}\label{obser:invariants}
All algorithm invariants hold for $R_{i+1}$.  
\end{observation}
\begin{proof}
Consider a node $u\in B_v$ for some $v\in R_i$. By definition, $|P_u|\leq |P_v|/2^{2^i}$. Since $|P_v|\leq 2n/2^{2^i}$, we obtain $|P_u|\leq 2n/2^{2^{i+1}}$. This proves the first invariant. 

For the second invariant, by definition, $P_u$ of all nodes $u\in B_v$ are edge disjoint (this is because the path from any leaf to the root of $\Psi(v)$ has exactly one border node). Since $P_v$ for all nodes $v\in R_i$ are edge disjoint, we obtain that $P_u$ for all nodes $u\in R_{i+1}$ are edge disjoint. This proves the second invariant. 

By definition, the parent of each node of $R_{i+1}$ is in the $i$-th super-level $\Psi_i$. Therefore, the third invariant also holds. 
\end{proof}

Finally, note that if $|P_v| = O(2^{2^i})$ for any $v \in R_i$, we simply set $\Psi_i(v) = \Psi(v)$, which marks the last super-level.  
Since $|P_v| \le 2n / 2^{2^i}$, this occurs for $i = O(\log\log n)$, implying that the total number of super-levels is $O(\log\log n)$.

\subparagraph{Preprocessing.}
With the super-levels defined above, we perform the following preprocessing.  

Consider the $i$-th super-level $\Psi_i$.  
For each node $v \in R_i$ and each node $u \in \Psi_i(v)$, note that $v$ is an ancestor of $u$ in $\Psi$.  
By the gate property, $P_u$ is separated from $P_v \setminus P_u$ by the gate $d_u$ of $u$.  
We construct a data structure $\calD_u(v)$ for $P_u$ with respect to $d_u$ in $P_v$, such that, for any query point $q \in P_u$, $\vis(P_v \setminus P_u, q)$ can be computed in $O(\log n)$ time.  
This is a diagonal-separated subproblem and such a data structure can be constructed in $O(|P_v|^2)$ space and $O(|P_v|^2 \log n)$ preprocessing time~\cite{ref:AronovVi02}, as described in Section~\ref{sec:diasep}.  

We now analyze the total space required for the data structures in $\Psi_i$.  
By the same argument as in Lemma~\ref{lem:internalnumber}, we have $|\Psi_i(v)| = O(2^{2^i})$ since $\Psi_i(v)$ consists of the internal nodes of $\Psi_b(v)$.  
Hence, the total space for all data structures $\calD_u(v)$, over all nodes $u \in \Psi_i(v)$ and all $v \in R_i$, is
$O(\sum_{v \in R_i} |P_v|^2 \cdot 2^{2^i})$.
Because the subpolygons $P_v$ for all $v \in R_i$ are edge-disjoint, we have $\sum_{v \in R_i} |P_v| = O(n)$.  
Moreover, since $|P_v| \le 2n / 2^{2^i}$, it follows that
$\sum_{v\in R_i}|P_v|^2=O(n^2/2^{2^i})$.
Therefore, the total space of all data structures in the $i$-th super-level is $O(n^2)$, and the total preprocessing time is $O(n^2 \log n)$.  

Hence, building the data structures for all $O(\log\log n)$ super-levels takes $O(n^2 \log\log n)$ space and $O(n^2 \log n \log\log n)$ preprocessing time.  
Including the preprocessing for Theorem~\ref{theo:treequeries}, the total space becomes $O(n^2 \log^4 n)$ and the total preprocessing time is $O(n^2 \log^5 n)$.


\subparagraph{Queries.}
Given a query point $q$, by a point location query, we find the node $u$ of $\Psi$ whose core $\core(u)$ contains $q$.  
We then compute $\vis(P_u, q)$ in $O(\log n)$ time using Theorem~\ref{theo:treequeries}.  

Assume that $u$ belongs to the $i$-th super-level for some $i$, i.e., $u \in \Psi_i$.  
Let $v$ be the node in $R_i$ such that $u \in \Psi_i(v)$.  
Using the data structure $\calD_u(v)$, we compute $\vis(P_v \setminus P_u, q)$ in $O(\log n)$ time.  
Since $P_u$ is separated from $P_v \setminus P_u$ by the gate $d_u$, we can obtain $\vis(P_v, q)$ by merging $\vis(P_u, q)$ and $\vis(P_v \setminus P_u, q)$ in $O(\log n)$ time using Lemma~\ref{lem:merge}.  

If $i = 0$, then $P_v = P$, and thus $\vis(q) = \vis(P_v, q)$ is already computed.  
Otherwise, let $u'$ be the parent of $v$.  
By the algorithm invariants, $u'$ lies in the $(i-1)$-th super-level.  
By Theorem~\ref{theo:treequeries}, we can compute $\vis(P_{u'} \setminus P_v, q)$ in $O(\log n)$ time.  
Since $P_v$ is separated from $P_{u'} \setminus P_v$ by the gate $d_v$, we can merge $\vis(P_v, q)$ and $\vis(P_{u'} \setminus P_v, q)$ in $O(\log n)$ time to obtain $\vis(P_{u'}, q)$.  
As $u'$ lies in the $(i-1)$-th super-level, we repeat the same procedure until reaching the $0$-th super-level, after which $\vis(q)$ is obtained.  

Since there are $O(\log\log n)$ super-levels and each super-level requires $O(\log n)$ time, the total query time is $O(\log n \log\log n)$.  
For reference, we summarize the result below.

\begin{lemma}\label{lem:prealgo}
We can construct a data structure of size $O(n^2 \log^4 n)$ in $O(n^2 \log^5 n)$ time for $P$ such that $\vis(q)$ can be computed in $O(\log n \log\log n)$ time for any query point $q \in P$.  
\end{lemma}

We further reduce the preprocessing complexities in the following theorem. 

\begin{theorem}\label{theo:quadraticspace}
Given a simple polygon $P$ with $n$ vertices, a data structure of size $O(n^2)$ can be constructed in $O(n^2 \log n)$ time such that the visibility polygon $\vis(q)$ can be computed in $O(\log n \log\log n)$ time for any query point $q \in P$.  
\end{theorem}
\begin{proof}
We follow the framework of Section~\ref{sec:bdscheme} to construct the balanced decomposition $\calT(P)$ and solve the diagonal-separated subproblems using the quadratic-space method of~\cite{ref:AronovVi02}, as reviewed in Section~\ref{sec:diasep}.  
This requires $O(n^2)$ space and $O(n^2 \log n)$ preprocessing time.  

Let $\calT_t(P)$ denote the subtree of $\calT(P)$ consisting of its top $t$ levels, where $t$ is the smallest integer such that the subpolygon $P_v$ at every node $v$ of the $t$-th level has size at most $n / \log^5 n$.  
By the definition of $\calT(P)$, we have $t = O(\log\log n)$.  

For each leaf $v \in \calT_t(P)$, we construct a data structure $\calD_v$ as in Lemma~\ref{lem:prealgo} for $P_v$.  
Each such data structure requires $O(|P_v|^2 \log^4 n)$ space and $O(|P_v|^2 \log^5 n)$ preprocessing time.  
Hence, the total preprocessing time for all data structures $\calD_v$ of the leaves $v \in \calT_t(P)$ is
$O(\sum_{v\in L}|P_v|^2\log^5 n)$,
where $L$ denotes the set of leaves of $\calT_t(P)$.  
Since the subpolygons $P_v$ for all leaves $v \in L$ form a partition of $P$, we have $\sum_{v \in L} |P_v| = O(n)$.  
Furthermore, because $|P_v| \le n / \log^5 n$ for all $v \in L$, it follows that
$\sum_{v\in L}|P_v|^2=O(n^2/\log^5 n)$.
Thus, the total preprocessing time of all $\calD_v$ is $O(n^2)$, and the total space is also $O(n^2)$.  

This completes the preprocessing, which requires $O(n^2)$ space and $O(n^2 \log n)$ time in total.  

For a query point $q \in P$, we first locate the leaf node $v \in L$ such that $P_v$ contains $q$, which takes $O(\log n)$ time using a point location query.  
Then, by Lemma~\ref{lem:prealgo}, we compute $\vis(P_v, q)$ in $O(\log n \log\log n)$ time using $\calD_v$.  
Next, as described in Section~\ref{sec:bdscheme}, we use the diagonal-separated data structure to compute $\vis(P_u, q)$ in $O(\log n)$ time for the parent $u$ of $v$.  
In general, by following the path from $v$ to the root in $\calT(P)$, $\vis(q)$ can be computed in an additional $O(\log n \log\log n)$ time since the path has $O(\log\log n)$ nodes.  
Therefore, the total query time for computing $\vis(q)$ is $O(\log n \log\log n)$.  
\end{proof}



\section{The boundary case}
\label{sec:Boundary_queries_only}

In this section, we present our results for the boundary case in which the query point $q$ is required to be on the boundary of $P$. We give two data structures, one based on the algorithmic framework in Section~\ref{sec:bdscheme}, and the other based on our new polygon decomposition $\Psi$ and following a similar framework to Section~\ref{sec:Optimal_query_time}. Both results rely on solving a diagonal-separated subproblem which we discuss first in the following. 

\subsection{A diagonal-separated subproblem}
\label{sec:boundarydiagonal}
Let $d$ be a diagonal that divides $P$ into two subpolygons $P_\ell$ and $P_r$. Without loss of generality, we assume that $d$ is vertical and $P_\ell$ is locally left of $d$. 
Let $\partial P_r=P_r\cap \partial P$, and similarly, let $\partial P_\ell=P_\ell\cap \partial P$. The diagonal-separated subproblem is to compute $\vis(P_\ell,q)$ for any query point $q\in \partial P_r$. 

Note that the two endpoints of $d$ are also the endpoints of $\partial P_r$. We will consider $\partial P_r$ a one-dimensional domain. 
For each element $w$ (i.e., a vertex or an edge) of $P_\ell$, if the cone $C_w$ (with respect to $d$, as defined in Section~\ref{sec:diasep}) is not empty, then its two bounding rays will shoot inside $P_r$.  
These two rays will hit two points of $\partial P_r$ and we let $I_w$ denote the interval of $\partial P_r$ between the two hit points. The observation is that for any $q\in \partial P_r$, $q$ is visible to $w$ only if $q\in I_w$. 

Consider a point $q\in \partial P_r$. Let $W_q$ be the set of elements $w$ of $P_\ell$ whose intervals $I_w$ contain $q$. We sort the elements of $W_q$ by their index order and store them in a binary search tree $T_q$. 
With $T_q$, in light of Observation~\ref{obs:boundaryorder},  $\vis(P_\ell,q)$ consists of the elements of $T_q$ clipped by the cone $C_q(s)$, where $s=\vis(d,q)$. More specifically, let $w_1(q)$ and $w_2(q)$ be the two points of $\partial P_\ell$ hit by the two bounding rays of $C_q(s)$, respectively. Then, an element $w$ of $T_q$ is visible to $q$ if and only if $w$ is between $w_1(q)$ and $w_2(q)$ in $T_q$. Hence, the portion of $T_q$ between $w_1(q)$ and $w_2(q)$ represents the combinatorial representation of $\vis(P_\ell,q)$. 

\subparagraph{Preprocessing.}
Based on the above discussion, we perform the following preprocessing work. We compute the interval $I_w$ for all elements $w\in P_\ell$. This can be done in $O(n\log n)$ time by first computing the cones $C_w$ of all elements $w\in \partial P_\ell$ in $O(n)$ time~\cite{ref:GuibasLi87} and then solving $O(n)$ ray-shooting queries in $O(n\log n)$ time with the help of a ray-shooting data structure~\cite{ref:ChazelleRa94,ref:HershbergerA95}. We sort all interval endpoints on $\partial P_r$, which together partition $\partial P_r$ into $O(n)$ intervals, called {\em basic intervals}. Let $\calI$ denote the set of all basic intervals. 

For each basic interval $I\in \calI$, $T_q$ is the same for all $q\in I$. Hence, we can maintain a single tree $T_I$ for each $I$. In addition, for two adjacent intervals $I$ and $I'$, $T_I$ and $T_{I'}$ differ by at most one element. To save space, as in the quadratic-space solution of \cite{ref:AronovVi02} described in Section~\ref{sec:diasep}, we use a persistent binary tree $T_{\calI}$ to store $T_I$ for all $O(n)$ basic intervals in  $O(n)$ space and $O(n\log n)$ time~\cite{ref:SnarnakPl86,ref:DriscollMa89}. In addition, we construct the GH data structure~\cite{ref:GuibasOp89} and a ray-shooting data structure~\cite{ref:ChazelleRa94,ref:HershbergerA95} for $P$, which take $O(n)$ time in total. 
This finishes our preprocessing. The total space is $O(n)$ and the total preprocessing time is $O(n\log n)$. 

\subparagraph{Queries.}
Given a query point $q\in \partial P_r$, we first find the interval $I\in \calI$ that contains $q$ in $O(\log n)$ time by binary search. With $I$, we can obtain the tree $T_q$ using the persistent tree $T_{\calI}$. Next, using the GH data structure, we compute the cone $C_q(s)$ with $s=\vis(d,q)$ in $O(\log n)$ time. Then, using two ray-shooting queries for the two bounding rays of $C_q(s)$, we compute $w_1(q)$ and $w_2(q)$ in $O(\log n)$. Finally, we conduct the clipping operation on $T_q$ to obtain a tree representing $\vis(q)$. The total query time is thus $O(\log n)$. 
\bigskip

In summary, we obtain the following result. 

\begin{lemma}\label{lem:boundarydiagonal}
Given a simple polygon $P$ of $n$ vertices and a diagonal $d$ dividing $P$ into two subpolygons $P_\ell$ and $P_r$, 
a data structure of $O(n)$ space can be constructed in $O(n\log n)$ time such that $\vis(P_\ell,q)$ can be computed in $O(\log n)$ time for any query point $q\in \partial P_r$. 
\end{lemma}

\subsection{The general problem}

We now consider the general problem of computing $\vis(q)$ for any query point $q \in \partial P$.  
We present two solutions.  

In our first solution, we follow the algorithmic framework described in Section~\ref{sec:bdscheme}, applying Lemma~\ref{lem:boundarydiagonal}, and obtain the following result (as already mentioned at the end of Section~\ref{sec:diasep}).  

\begin{theorem}\label{theo:boundaryqueries}
Given a simple polygon $P$ of $n$ vertices, a data structure of $O(n \log n)$ space can be constructed in $O(n \log^2 n)$ time such that the visibility polygon $\vis(q)$ can be computed in $O(\log^2 n)$ time for any query point $q \in \partial P$.
\end{theorem}

Our second data structure employs our new decomposition $\Psi$ and follows the same framework as the method in Section~\ref{sec:Optimal_query_time}, using the boundary-case diagonal-separated data structure in Lemma~\ref{lem:boundarydiagonal}.  
We briefly discuss it below, highlighting the modifications relative to the algorithm in Section~\ref{sec:Optimal_query_time}.

\subparagraph{A data structure for geodesic triangles.}
The method in Section~\ref{sec:Optimal_query_time} relies on Theorem~\ref{theo:treequeries}, which in turn is built upon a data structure for geodesic triangles (i.e., Lemma~\ref{lem:geotriangle}).  
We therefore need to handle the boundary case of this problem.  
Since the query points are on $\partial P$, we only need to consider the second subproblem in Lemma~\ref{lem:geotriangle}, that is, when $q \in P_d$ for some diagonal $d$ on the boundary of $\triangle$.  

The second subproblem is solved in Section~\ref{sec:secondsub}.  
We follow the notation there, except that now $q$ is restricted to lie on $\partial P \cap P_d$.  
We design a data structure for query points in the pockets of the diagonals on $\pi$, i.e., the side of $\triangle$ containing $d$.  
Intuitively, since $\pi$ is a concave chain, we extend the algorithm in Section~\ref{sec:boundarydiagonal} for the diagonal-separated problem to this {\em concave-chain-separated} problem, i.e., $\pi$ now plays the role of the separating diagonal in the problem of Section~\ref{sec:boundarydiagonal}.  

Specifically, let $P_r$ denote the union of the pockets of the diagonals along $\pi$, and let $P_\ell = P \setminus P_r$.  
Define $\partial P_r = \partial P \cap P_r$, and define $\partial P_\ell$ similarly.  
Note that $\partial P_r$ is a portion of $\partial P$ whose endpoints are those of $\pi$, so $\partial P_r$ forms a one-dimensional polygonal chain; the same holds for $\partial P_\ell$.  

For each element $w$ (i.e., a vertex or an edge) of $\partial P_\ell$, we compute the portion of $\pi$ visible to $w$, which is constrained by a cone $C_{\pi}(w)$.  
We can compute all such cones $C_{\pi}(w)$ for every element $w$ of $\partial P_\ell$ in $O(n)$ time by constructing the shortest path trees from the two endpoints of $\pi$, following the same approach as the algorithm of Guibas and Lubiw~\cite{ref:GuibasLi87} for computing visibility cones with respect to a single edge. Indeed, the same algorithm applies here because $\pi$ is a concave chain.  
Using $C_{\pi}(w)$, we define an interval on $\partial P_r$ by the two points where the bounding rays of $C_{\pi}(w)$ hit $\partial P_r$.  
The remainder of the algorithm proceeds as in Section~\ref{sec:boundarydiagonal}.  
The data structure thus requires $O(n)$ space and $O(n \log n)$ preprocessing time.  

Given a query point $q \in \partial P_r$, suppose $q$ lies in the pocket of a diagonal $d$ on $\pi$.  
We compute the cone $C_q(s)$ with $s = \vis(d, q)$ and obtain the tree $T_q$ as defined earlier.  
We then use $C_q(s)$ to clip $T_q$, yielding a binary tree representing  $\vis(P_\ell, q)$, which is $\vis(P \setminus P_d, q)$ since $\pi$ is a concave chain.  
The total query time is $O(\log n)$.  

The above constructs a data structure for $\pi$, i.e., for query points $q$ in the pockets of diagonals along $\pi$.  
We also build the same data structures for the other two sides of $\triangle$.  
All together, these require $O(n)$ space and $O(n \log n)$ preprocessing time.  
Hence, for any query point $q \in \partial P \cap P_d$ for a diagonal $d \in \partial \triangle$, $\vis(P \setminus P_d, q)$ can be computed in $O(\log n)$ time.

\subparagraph{The main data structure.}
We now use the above geodesic-triangle data structure to construct the data structure of Theorem~\ref{theo:treequeries} (following the method in Section~\ref{sec:theo:treequeries}) and obtain its boundary version.  
The query time remains the same, but the space becomes $O(n \log^3 n)$ and the preprocessing time becomes $O(n \log^4 n)$.  
Indeed, by Lemma~\ref{lem:Abound} and following the analysis in Section~\ref{sec:theo:treequeries}, the total space of the data structures in each level of $\Psi$ is $O(n \log^2 n)$.  
Since $\Psi$ has $O(\log n)$ levels, the total space of the full data structure is $O(n \log^3 n)$, and the preprocessing time gains an additional logarithmic factor.  

We now plug this boundary version of Theorem~\ref{theo:treequeries} into the method of Section~\ref{sec:Optimal_query_time}.  
In addition, when constructing the data structure $\calD_v$ with respect to the gate $d_v$ for each internal node $v \in \Psi_b$, we use Lemma~\ref{lem:boundarydiagonal} instead.  
The resulting recurrence for the space $S(n)$ is as follows:
\[
S(n) = \sum_{v \in B(\Psi)} S(|P_v|) + O(nr),
\]
which solves to $O(n^{1+\epsilon})$ by setting $r = n^{\epsilon}$.  
Similarly, the preprocessing time is also $O(n^{1+\epsilon})$.  
Including the boundary version of Theorem~\ref{theo:treequeries}, the total space and preprocessing time of the overall data structure are both $O(n^{1+\epsilon})$, while the query time remains $O(\log n)$.  

In summary, we obtain the following result for the boundary case.

\begin{theorem}
Given a simple polygon $P$ of $n$ vertices, a data structure can be constructed in $O(n^{1+\epsilon})$ space and preprocessing time, for any constant $\epsilon > 0$, such that the visibility polygon $\vis(q)$ can be computed in $O(\log n)$ time for any query point $q \in \partial P$. 
\end{theorem}

Note that applying the approach from Section~\ref{sec:A_Quadratic_Space_Structure} would not yield a better result, since that method is effective only when the preprocessing of the diagonal-separated subproblem requires $\Omega(n^2)$ complexity.

\subparagraph{The vertex case.}
Finally, if the query point $q$ is a vertex of $P$, then Theorem~\ref{theo:boundaryqueries} immediately yields the following result.

\begin{corollary}
Given a simple polygon $P$ with $n$ vertices, one can construct a data structure of $O(n \log^2 n)$ space in $O(n \log^2 n)$ time such that, for any query point $q$ that is a vertex of $P$, the visibility polygon $\vis(q)$ can be reported in $O(|\vis(q)|)$ time.
\end{corollary}

\begin{proof}
During preprocessing, we construct the data structure of Theorem~\ref{theo:boundaryqueries} in $O(n \log n)$ space and $O(n \log^2 n)$ time; let $\calD$ denote this data structure.

Let $S$ be the set of vertices $v$ of $P$ for which $k_v \leq \log^2 n$, where $k_v$ is the number of vertices of $\vis(v)$. We can compute $S$ in $O(n \log^2 n)$ time as follows. For each vertex $v$, we use $\calD$ to compute a binary search tree representing $\vis(v)$ in $O(\log^2 n)$ time by Theorem~\ref{theo:boundaryqueries}. We then traverse the tree and output vertices of $\vis(v)$ until either the entire visibility polygon has been reported or more than $\log^2 n$ vertices have been produced. In the former case, we conclude that $k_v \leq \log^2 n$ and add $v$ to $S$. In the latter case, $v\not\in S$. 
Thus, determining whether a given vertex belongs to $S$ takes $O(\log^2 n)$ time, and processing all vertices takes $O(n \log^2 n)$ time.

For each vertex $v \in S$, we explicitly store $\vis(v)$ in the data structure. Since each such visibility polygon has size  $O(\log^2 n)$ and there are at most $n$ vertices, this requires $O(n \log^2 n)$ space and construction time in total.

This completes the preprocessing, which uses $O(n \log^2 n)$ space and time.

Now consider a query vertex $q$. If $q \in S$, then since $\vis(q)$ is explicitly stored, it can be reported in $O(|\vis(q)|)$ time. Otherwise, we know that $|\vis(q)| =\Omega(\log^2 n)$, and we compute $\vis(q)$ using $\calD$ in $O(\log^2 n + |\vis(q)|)$ time by Theorem~\ref{theo:boundaryqueries}. Since $|\vis(q)| = \Omega(\log^2 n)$, this is $O(|\vis(q)|)$ time.
\end{proof}

\bibliography{references}




\end{document}